\documentclass[final,twocolumn]{IEEEtran}

\usepackage{amsmath,amssymb,amsthm}
\usepackage{xifthen}
\usepackage{mathtools}
\usepackage{enumerate}
\usepackage{microtype}
\usepackage{xspace}
\usepackage{bm}
\usepackage[T1]{fontenc}
\usepackage{fancyhdr}
\usepackage{lastpage}
\usepackage{bbm}
\usepackage{cases}

\newcommand{\mbs}[1]{\pmb{#1}}
\newcommand{\vect}[1]{{\lowercase{\mbs{#1}}}}
\newcommand{\mat}[1]{{\uppercase{\mbs{#1}}}}

\newcommand{\T}{{\scriptscriptstyle\mathsf{T}}}
\renewcommand{\H}{{\scriptscriptstyle\mathsf{H}}}

\newcommand{\cond}{\,\vert\,}
\renewcommand{\Re}[1][]{\ifthenelse{\isempty{#1}}{\operatorname{Re}}{\operatorname{Re}\left(#1\right)}}
\renewcommand{\Im}[1][]{\ifthenelse{\isempty{#1}}{\operatorname{Im}}{\operatorname{Im}\left(#1\right)}}

\newcommand{\sv}{\vect{s}}

\newcommand{\vv}{\vect{v}}
\newcommand{\wv}{\vect{w}}

\newcommand{\Am}{\mat{a}}
\newcommand{\Bm}{\mat{b}}

\newcommand{\Hm}{\mat{h}}
\newcommand{\Jm}{\mat{j}}

\newcommand{\Mm}{\mat{M}}

\newcommand{\Pm}{\mat{p}}

\newcommand{\Id}{\mat{\mathrm{I}}}

\newcommand{\CN}[1][]{\ifthenelse{\isempty{#1}}{\mathcal{N}_{\mathbb{C}}}{\mathcal{N}_{\mathbb{C}}\left(#1\right)}}

\renewcommand{\P}[1][]{\ifthenelse{\isempty{#1}}{\mathbb{P}}{\mathbb{P}\left(#1\right)}}
\newcommand{\E}[1][]{\ifthenelse{\isempty{#1}}{\mathbb{E}}{\mathbb{E}\left[#1\right]}}
\newcommand{\I}[1][]{\ifthenelse{\isempty{#1}}{\mathbb{I}}{\mathbb{I}\left\{#1\right\}}}
\renewcommand{\det}[1][]{\ifthenelse{\isempty{#1}}{\mathrm{det}}{\text{det}\left(#1\right)}}
\newcommand{\trace}[1][]{\ifthenelse{\isempty{#1}}{\mathrm{tr}}{\text{tr}\left(#1\right)}}
\newcommand{\rank}[1][]{\ifthenelse{\isempty{#1}}{\mathrm{rank}}{\text{rank}\left(#1\right)}}
\newcommand{\diag}[1][]{\ifthenelse{\isempty{#1}}{\mathrm{diag}}{\text{diag}\left(#1\right)}}
\newcommand{\Cov}[1][]{\ifthenelse{\isempty{#1}}{\mathsf{Cov}}{\mathsf{Cov}\left(#1\right)}}

\newcommand{\defeq}{\triangleq}

\newtheorem{proposition}{Proposition}
\newtheorem{remark}{Remark}[section]
\newtheorem{definition}{Definition}
\newtheorem{theorem}{Theorem}
\newtheorem{lemma}{Lemma}

\newcommand{\const}{c_0}
\newcommand{\constH}{c_H}

\newcommand{\rvVec}[1]{\pmb{\mathrm{#1}}}
\newcommand{\rvMat}[1]{\pmb{\mathsf{#1}}}
\newcommand{\Hmax}{\lambda_{\pmb{H}}}

\newcommand{\nt}{n_{\text{t}}}
\newcommand{\nr}{n_{\text{r}}}

\mathtoolsset{showonlyrefs}

\title{On the Multiplexing Gain of Discrete-Time MIMO Phase Noise Channels} 
\author{
Sheng~Yang,~\IEEEmembership{Member,~IEEE,}
Shlomo Shamai~(Shitz),~\IEEEmembership{Fellow,~IEEE} 
\thanks{S. Yang is with LSS, CentraleSupélec,  3 rue Joliot-Curie, 91190
Gif-sur-Yvette, France.~(e-mail: \texttt{sheng.yang@centralesupelec.fr})}
\thanks{S. Shamai~(Shitz) is with Technion-Israel Institute of Technology, Haifa, Israel.~(e-mail: \texttt{sshlomo@ee.technion.ac.il})}
\thanks{The work of S. Shamai was supported by the Israel Science Foundation (ISF). The material in this paper
was presented in part at the 2016 IEEE Information Theory Workshop.}
}

\begin{document}

\maketitle
\begin{abstract}
  The capacity of a point-to-point discrete-time multi-input-multiple-output~(MIMO) channel with
  phase uncertainty~(MIMO phase noise channel) is still open. As a matter of fact, even
  the pre-log~(multiplexing gain) of the capacity in the high signal-to-noise
  ratio~(SNR) regime is unknown in general. We make some progresses in this direction
  for two classes of such channels. With phase noise on the individual
  paths of the channel~(model~A), we show that the multiplexing gain is $\frac{1}{2}$,
  which implies that the capacity \emph{does not} scale with the channel dimension at
  high SNR. With phase noise at both the input and output of the
  channel~(model~B), the multiplexing gain is upper-bounded by
  $\frac{1}{2} \min\{{\nt},(\nr-2)^+\! + 1\}$, and lower-bounded by $\frac{1}{2} \min\{\nt,
  \lfloor \frac{\nr+1}{2} \rfloor\}$, where $\nt$ and $\nr$ are the number of transmit and
  receive antennas, respectively. The multiplexing gain is enhanced to
  $\frac{1}{2}\min\{\nt,\nr\}$ without receive phase noise, and to
  $\frac{1}{2}\min\{2\nt-1,\nr\}$ without transmit phase noise. In all the cases of
  model~B, the multiplexing gain scales linearly
  with $\min\{\nt,\nr\}$.
  Our main results rely on the derivation of non-trivial upper and lower bounds on the capacity of such channels. 
\end{abstract}


\section{Introduction}

The capacity of a point-to-point multiple-input-multiple-output~(MIMO) Gaussian channel is well
known in the coherent case, i.e., when the channel state information is available at the
receiver~\cite{Foschini,Telatar}. 
The capacity of the noncoherent MIMO channels, however, is still open in general.
Nevertheless, asymptotic results of such channels, e.g., at high
signal-to-noise ratio~(SNR), have been obtained in some important cases.   

In the seminal paper~\cite{Moser}, Lapidoth and Moser proposed a powerful technique, called the
duality approach, that can be applied to a large class of fading channels and derived the exact
high SNR capacity  up to an $o(1)$ term. In particular, when the differential entropy of the
channel matrix is finite, i.e., $h(\rvMat{H})>-\infty$, it was shown in \cite{Moser} that the
pre-log (a.k.a. multiplexing gain), of the capacity is $0$ and the high-SNR capacity is $\log\log
\mathsf{SNR} + \chi(\rvMat{H}) + o(1)$ where $\chi(\rvMat{H})$ is the so-called fading number of
the channel. In addition, capacity upper and lower bounds for the MIMO Rayleigh and Ricean
channels were obtained and shown to be tight at both low and high SNR regimes. In~\cite{Zheng_Tse}, Zheng
and Tse showed that for noncoherent block fading MIMO Rayleigh channels with with coherence
time~$T$, the pre-log is $M^*(1-M^*/T)$ where
$M^* \defeq \min\left\{ \nt, \nr, \lfloor \frac{T}{2} \rfloor
\right\}$ with $\nt$ and $\nr$ being the number of transmit and receive antennas, respectively. In this work, we are interested in the MIMO phase noise channels in which the phases of
the channel coefficients are not perfectly known. 

Applying the duality approach and the ``escape-to-infinity'' property of the channel input,
Lapidoth characterized the high-SNR capacity of the discrete-time phase noise channel in the
single-antenna case~\cite{Lapidoth-ITW02}. It was shown in \cite{Katz} that the
capacity-achieving input distribution is in fact discrete. Recently, capacity upper and lower
bounds of the single-antenna channels with Wiener phase noise have been extensively studied in
the context of optical fiber and microwave communications~(see \cite{Barletta2012, Ghozlan2015,
Durisi-capa} and the references therein). In these works, the upper bounds are derived via
duality and lower bounds are computed numerically using the auxiliary channel technique proposed
in~\cite{Arnold2006}. In particular, in \cite{Durisi-capa}, Durisi \emph{et~al.} investigated
the MIMO phase noise channel with a common phase noise, a scenario motivated by the microwave
link with centralized oscillators. The SIMO and MISO channels with common and separate phase
noises are considered in \cite{Khanzadi2015}. The  $2\times2$ MIMO phase
noise channel with independent transmit and receive phase noises at each antenna was studied in
\cite{Durisi2013}, where the authors showed that the multiplexing gain is $\frac{1}{2}$ for a
specific class of input distributions. For general MIMO channels with separate phase noises,
estimation and detection algorithms have been proposed in \cite{Nasir, Tanumay}. However  
for such channels, even the multiplexing gain is unknown, to the best of our knowledge. 

In this work, we make some progresses in this direction.
We consider two classes of discrete-time stationary and ergodic MIMO
phase noise channels: model~A with individual phase noises on the entries of the channel matrix,
and model~B with individual phase noises at the input and the output of the channel instead.  The phase
noise processes in both models are assumed to have finite differential entropy rate.  For
model~A, we obtain the exact multiplexing gain $\frac{1}{2}$ for any channel dimension, which
implies that the capacity \emph{does not} scale with the channel dimension at high SNR.  For
model~B with both transmit and receive phase noises, we show that the multiplexing gain is
upper-bounded by $\frac{1}{2} \min\{{\nt},(\nr-2)^+\!
+ 1\}$, and lower-bounded by $\frac{1}{2} \min\{\nt, \lfloor \frac{\nr+1}{2} \rfloor\}$. The upper and lower bounds
coincide for $\nr\le3$ or $\nr\ge2\nt-1$. Further, when receive phase noise
is absent, the multiplexing gain is improved and we obtain the exact value of
$\frac{1}{2}\min\{\nt,\nr\}$. If the transmit phase noise is absent instead, the multiplexing
gain becomes  $\frac{1}{2}\min\{2\nt-1,\nr\}$.

The main technical contribution of this paper is two-fold. First, we derive a non-trivial upper bound
on the capacity of the MIMO phase noise channel with separate phase noises. The novelty of
the upper bound lies in the finding of a suitable auxiliary distributions with which we apply the
duality upper bound~\cite{Topsoe,Kemperman,Moser}. It is worth mentioning that,
the class of single-variate Gamma output distributions, as the essential ingredient that led to the tight capacity upper bounds on previously studied channels,
 are not suitable for MIMO phase noise
 channels in general. In this paper, we introduce a class of
\emph{multi}-variate Gamma distributions that, combined with the duality upper bound, allows us to
obtain a complete pre-log characterization for model~A and partially for model~B. The second
contribution is the derivation of the capacity lower bounds for model~B, based on the remarkable
property of the differential entropy of the output vector in this channel. Namely, we prove that, at
high SNR, the pre-log of the said entropy can go beyond the rank of the channel matrix, $\min\left\{
\nt,\nr \right\}$, and scales as $\nr \log\mathsf{SNR}$ as long as $\nr\le2\nt-1$. The upper and lower bounds
suggest that, with $\nr\ge 2\nt-1$ receive antennas, $\nt$ transmitted real symbols can be recovered at
high SNR. This result has an interesting interpretation based on dimension counting. Let us consider
the example of independent and memoryless transmit and receive phase noises uniformly distributed in
$[0,2\pi)$. In this case, phases of the input and the output do not contain any useful information,
only the amplitudes matter. Note that the $\nr$ output amplitudes are (non-linear)
equations of $2\nt-1$
unknowns, namely, the $\nt$ input amplitudes and the $\nt-1$ relative input phases, assuming the additive
noises are negligible at high SNR. It is now not too hard to believe that with $\nr=2\nt-1$ equations,
the receiver can successfully decode the $\nt$ input amplitudes by
solving the equations. This is however not possible with $\nr<2\nt-1$, in which case there are
too many unknowns as compared to the number of equations.
 Nonetheless, we can reduce the number of active transmit antennas to $\nt'<\nt$ such that
 $2\nt'-1\le \nr$, which means that the achievable multiplexing gain is $\frac{\nt'}{2} \le
 \frac{1}{2} \lfloor \frac{\nr+1}{2} \rfloor$. A formal proof in Section~\ref{sec:LB}
 validates such an argument.

The remainder of the paper is organized as follows. The system model and main results are
presented in Section~\ref{sec:model}. Some preliminaries useful for the proof of the main
results are provided in Section~\ref{sec:preliminaries}. The upper bounds are derived in
Section~\ref{sec:upperbound} and Section~\ref{sec:upperbound-B}. We prove the lower bound for model~B
in section~\ref{sec:LB}. Concluding remarks are given in Section~\ref{sec:conclusion}. Most of the
proofs are presented in the main body of the paper, with some details deferred to the Appendix.

\section{System Model and Main Results}
\label{sec:model}

\subsection*{Notation}
Throughout the paper, we use the following notational conventions. For random quantities, we use
upper case letters, e.g., $X$, for scalars, upper case letters with bold and non-italic fonts,
e.g., $\rvVec{V}$, for vectors, and upper case letter with bold and sans serif fonts, e.g.,
$\rvMat{M}$, for matrices.  Deterministic quantities are denoted in a rather conventional way
with italic letters, e.g., a scalar $x$, a vector $\pmb{v}$, and a matrix $\pmb{M}$. Logarithms
are in base $2$.  The
Euclidean norm of a vector and a matrix is denoted by $\|\vv\|$ and
$\|\Mm\|$, respectively. The transpose and conjugated transpose of $\Mm$ are $\Mm^\T$ and
$\Mm^\H$, respectively. $\pmb{\Hm}^\dag$ is the pseudo-inverse of a
\emph{tall} matrix $\pmb{\Hm}$. The argument~(phase)
of a complex value $x$ is denoted by $\angle_x\in[0,2\pi)$.  
We use $\Am \circ \Bm$ to denote the Hadamard~(point-wise) product between vectors/matrices. 
$x_{n+1}^{n+k}$ is a $k$-tuple or a column vector of $(x_{n+1},\ldots,x_{n+k})$; for brevity
sometimes $x^k$ replaces $x_1^k$. 
For convenience, wherever confusion is improbable, elementary scalar functions
applied to a vector, e.g., $|\pmb{x}|$ or $\cos(\pmb{\theta})$, stand for a point-wise map
on
each element of the vector, and return a vector with the same dimension as the argument. 
We use $(\theta)_{2\pi}$ to denote $(\theta\!\mod2\pi)$, and
$(x)^+ = \max\{x,0\}.$ $\Gamma(x)$ is the gamma function. 
We also use $\const$ to represent a bounded constant whose value is
irrelevant but may change at each occurrence. Similarly, $\constH$ is a constant that may depend
on $\pmb{H}$ but the value is irrelevant and bounded for almost all $\pmb{H}$. 

\begin{figure*}[!t]
  \centering
  \includegraphics[width=0.7\textwidth]{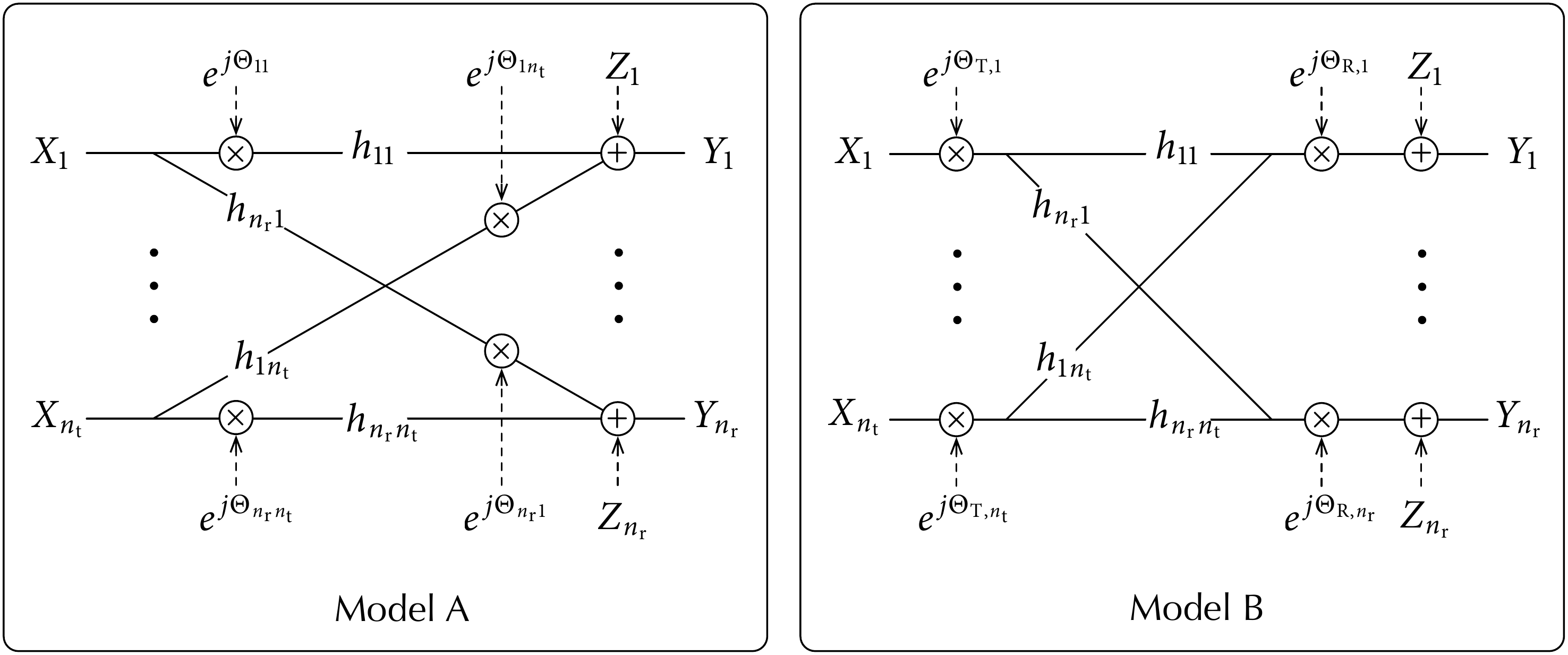}
  \caption{Two models considered in this work: model~A with path phase noise, and
  model~B with transmit/receive phase noise.}
  \label{fig:models}
\end{figure*}

\subsection{Channel model}

In this paper, we are interested in a class of
discrete-time MIMO phase noise channels with $\nt$ transmit
antennas and $\nr$ receive antennas, defined by
\begin{align}
  \rvVec{Y}_t &= (\pmb{H} \circ e^{j\rvMat{\Theta}_t}) \,
  \pmb{x}_t + \rvVec{Z}_t, \quad t=1,2,\ldots,N, \label{eq:model}
\end{align}%
where the deterministic channel matrix $\pmb{H}$ belongs to a set $\mathcal{H} \subset \mathbb{C}^{\nr \times \nt}$ of generic matrices\footnote{It
means that the channel matrix $\Hm$ does not lie on any algebraic hypersurface. If we draw the
entries of $\Hm$ i.i.d.~from a continuous distribution, then $\Hm$ is generic almost surely.};
$\pmb{x}_t\in\mathbb{C}^{\nt\times1}$ is the
input vector at time~$t$, with the
average power constraint $\frac{1}{N} \sum_{t=1}^N
\|\pmb{x}_t\|^2 \le P$;
the additive noise process $\{\rvVec{Z}_t\}$ is assumed to be spatially and temporally white with
$\rvVec{Z}_t \sim \mathcal{CN}(0, \Id_{\nr})$;
$\rvMat{\Theta}_t$ is the matrix of phase noises on the
individual entries of $\pmb{H}$ at time~$t$; the phase noise process
$\{\rvMat{\Theta}_t\}$ is stationary and ergodic,
and is independent of the additive noise process
$\{\rvVec{Z}_t\}$. Both $\{\rvVec{Z}_t\}$ and
$\{\rvVec{\Theta}_t\}$ are unknown to the transmitter and
the receiver. Since the additive noise power is
normalized, the transmit power $P$ is identified with the
SNR throughout the paper. 
The end-to-end channel is captured by the
random channel matrix $\rvMat{H} \defeq \bigl[h_{ik}
e^{\Theta_{ik}}\bigr]_{i,k}$.

In this paper, we consider two types of discrete-time phase noise processes\footnote{The limitation of discrete-time phase
noise model, which ignores the filtering before sampling in practical continuous-time communication
systems, has been discussed in \cite{Ghozlan2013b} and \cite{Ghozlan2015}.}
according to the spatial structures, as shown in Fig.~\ref{fig:models}: 
\begin{itemize}
  \item Model~A refers to channels with phase
    uncertainty on the individual paths~(path phase noise), such that the sequence~$\{\rvMat{\Theta}_t\}$ has finite entropy rate
    \begin{align}
      h(\{\rvMat{\Theta}_t\}) > -\infty.
    \end{align}%
    It corresponds to the case where the phase information of the channel cannot be
    obtained accurately, e.g., in optical fiber communications.  This model covers the
    channel with spatially independent phase noises as a special case.  
  \item Model B refers to channels with phase noises at the input and/or output, i.e.,
    $\Theta_{ik} = \Theta_{\text{R},i} +
    \Theta_{\text{T},k}$.
    The vector~$\rvVec{\Theta}_{\text{T}} \defeq
    \bigl[\Theta_{\text{T},i}\bigr]_{i=1}^{\nt}$ 
    contains the $\nt$ phase noises at the transmit
    antennas, and 
    $\rvVec{\Theta}_{\text{R}} \defeq
    \bigl[\Theta_{\text{R},k}\bigr]_{k=1}^{\nr}$ is the
    vector of the $\nr$ phase noises at the receive
    antennas. This model captures the phase corruption at both the
    transmit and receive RF chains, e.g., caused by
    imperfect oscillators. We consider three cases of model~B:
    \begin{itemize}[leftmargin=.5in]
      \item[B1)] with both transmit and receive phase noises such that     
        $h(\{\rvVec{\Theta}_{\text{T},t}, \rvVec{\Theta}_{\text{R},t}\})>-\infty$; 
      \item[B2)] with only transmit phase noise such that $h(\{\rvVec{\Theta}_{\text{T},t}\})>-\infty$;
      \item[B3)] with only receive phase noise such that $h(\{\rvVec{\Theta}_{\text{R},t}\})>-\infty$.
    \end{itemize}
    Note that model B1 covers the case where both the transmitter and receiver use separate
    (and imperfect) oscillators for different antennas, whereas models B2 and B3 contain the
    case with centralized oscillators at one side and separate oscillators at the other side. 
\end{itemize}

The capacity of such a stationary and ergodic channel is~\cite{Moser,Cover}
\begin{align}
  C(P) \defeq \lim_{N\to\infty} \sup \frac{1}{N} 
  I(\rvVec{X}^N; \rvVec{Y}^N), \label{eq:capa}
\end{align}%
where the supremum is taken over all distributions with the average power constraint 
\begin{align}
  \frac{1}{N} \sum_{k=1}^N \E \left[ \|\rvVec{X}_k\|^2 \right] &\le P. 
\end{align}%

Our work focuses on the multiplexing gain $r$ of
such a channel, defined as the pre-log of the
capacity~$C(P)$ as $P\to\infty$,
\begin{align}
  r &\defeq \lim_{P\to\infty} \frac{C(P)}{\log P}. 
\end{align}%

\subsection{Main results}

\begin{figure*}[t]
  \centering
  \includegraphics[width=0.8\textwidth]{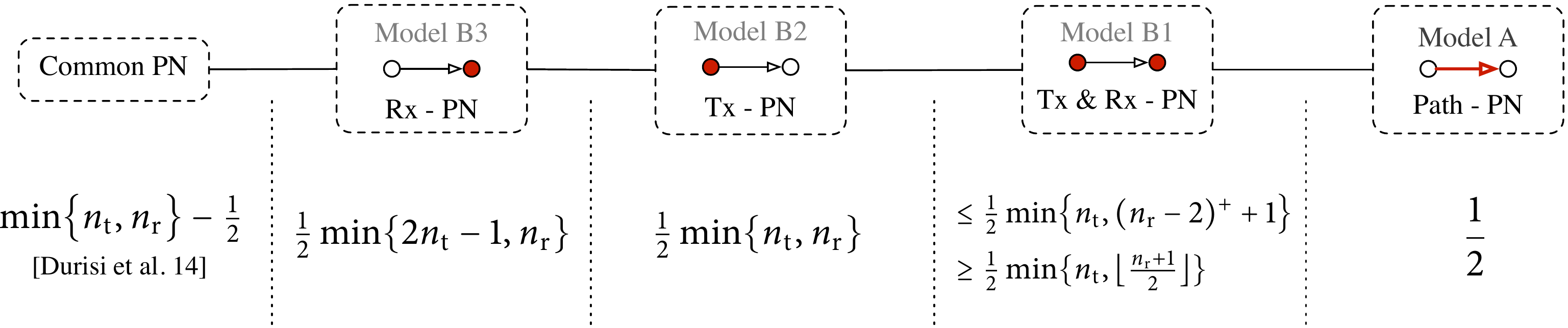}
  \caption{Multiplexing gain of the MIMO phase noise channels.}
  \label{fig:DoF}
\end{figure*}

The main results of this work are summarized as follows, and are illustrated in
Fig.~\ref{fig:DoF}. 
First, the case with common phase noise is rather straightforward from \cite{Durisi-capa}.  
\begin{proposition} \label{prop:1}
  With common phase noise, i.e., $\rvMat{\Theta}_{t} =
  \Theta_t\pmb{1}_{\nr \times \nt}$ and $h(\{\Theta_t\})>-\infty$, the multiplexing gain is
  $\min\{\nt,\nr\} - \frac{1}{2}$.
\end{proposition}
\begin{proof}
  The proof is provided in Appendix~\ref{app:prop1}. 
\end{proof}

Then our new results are on channels with separate phase noises either on the individual
paths~(model~A) or at the input/output~(model~B) of the channel. 
\begin{theorem}\label{thm:A}
  The multiplexing gain of model~A is $\frac{1}{2}$. 
\end{theorem}

The above result shows that extra transmit and receive antennas do not improve the
multiplexing gain of a channel with phase uncertainty on each path of the channel. The
achievability of the single-antenna case was shown in \cite{Lapidoth-ITW02}. Our main
contribution lies in the converse, as will be shown in Section~\ref{sec:upperbound}. 

\begin{theorem}\label{thm:B}
  The multiplexing gain of model~B is 
  \begin{itemize}
    \item upper-bounded by $\frac{1}{2}
      \min\{{\nt},(\nr-2)^+ + 1\}$, and lower-bounded by
      $\frac{1}{2} \min\{\nt, \lfloor \frac{\nr+1}{2}
      \rfloor\}$ with both transmit and receive phase noises,
      the upper bound is achievable when $\nr\le3$ or $\nr\ge 2\nt-1$;
    \item $\min\{\frac{\nr}{2}, \frac{\nt}{2}\}$, with only transmit phase noise;
    \item $\min\{\frac{\nr}{2}, \nt-\frac{1}{2}\}$, with only receive phase noise.
  \end{itemize}
\end{theorem}

Interestingly, the multiplexing gain of model~B depends on the number of transmit and receive antennas
differently, which is rarely the case for previously studied point-to-point MIMO channels. 
\begin{remark} As shown in Fig.~\ref{fig:DoF}, transmit phase noise is more detrimental
  than receive phase noise, and strictly so when $\nr>\nt>1$. Intuitively, with transmit phase
  noise each transmitted symbol is accompanied by a different phase noise symbol, which
  means that no more than half of the total spatial degrees of freedom is available for
  useful signal. On the other hand, with receive phase noise, although half of the received
  signal dimension is occupied by phase noises, it is enough to increase the number of
  receive antennas to recover almost all transmitted symbols.  \end{remark}

\begin{remark} Obviously, the multiplexing gain of model~B1 is upper-bounded by that of
  models~B2 and B3.  Such a ``trivial'' upper bound is given by
  $\min\{\frac{\nt}{2},\frac{\nr}{2}, \nt - \frac{1}{2}\} = \min\{\frac{\nt}{2},\frac{\nr}{2}\}$.
  When $\nr\le \nt$, the optimal multiplexing gain is $\frac{\nr}{2}$ with phase noises at either
  side of the channel, whereas no more than $\frac{(\nr-2)^++1}{2}$ is achievable with phase
  noises at both sides. These are the cases for which model~B1 is strictly ``worse'' than
  both models~B2 and B3.  When $\nr\ge 2\nt-1$, with transmit phase noise, the optimal
  multiplexing gain is $\frac{\nt}{2}$ regardless of the presence of receive phase noise.
\end{remark}

\begin{remark} Theorem~\ref{thm:B} shows that, when $\nt=\nr=2$ and $3$, the exact multiplexing
  gain of model~B1 is $\frac{(\nr-2)^++1}{2}$ which gives $\frac{1}{2}$ and $1$,
  respectively. In contrast, the trivial upper bound provides~$1$ and $\frac{3}{2}$,
  respectively. These are the two cases of model~B1 for which we obtain exact multiplexing
  gain that is strictly lower than that of models~B2 and~B3.  \end{remark}

The remainder of the paper is dedicated to the proof of the main results. We
start with some mathematical preliminaries.

\section{Preliminaries}
\label{sec:preliminaries}

\begin{definition}[Multivariate Gamma distribution \cite{Mathal}]
  \label{def:mv-Gamma}
  The $n$-variate Gamma distribution has the following density function
  \begin{align}
    p(\sv) &= {s_1^{\alpha_1-1} (s_2-s_1)^{\alpha_2-1}\cdots
    (s_n-s_{n-1})^{\alpha_n-1}} \nonumber \\
    &\qquad \cdot g_{\pmb{\alpha}}\, {\mu^{\alpha_1+\cdots+\alpha_n}}
    \exp(-\mu s_n),
  \end{align}%
  for $0 < s_1 < s_{2} < \cdots < s_n$, and 0 elsewhere,
  with $\mu>0$ and $\alpha_i > 0$, $i=1,\ldots,n$;
  $g_{\pmb{\alpha}}\defeq 1/\left(\prod_{i=1}^n \Gamma(\alpha_i)\right)$ is the normalization factor.  When
  $n=1$, we have the standard single-variate Gamma distribution,
  \begin{align}
    p(s) &= g_{{\alpha}}\, \mu^{\alpha} {s}^{\alpha-1} 
    \exp(-\mu s), \quad s>0,\ \alpha>0. \label{eq:sv-gamma}
  \end{align}%
\end{definition}

\begin{lemma}[Monotonicity of $\E_{\chi^2_k(\lambda)} {\bigl[\log
  X\bigr]} $ over $\lambda$] \label{lemma:chi2}
  Let ${X}$ be a non-central Chi-square distribution with degrees of freedom $k$ and noncentrality
  parameter $\lambda$, denoted as $X\sim\chi^2_{k}(\lambda)$. Then the expected logarithm of
  $X$ is strictly increasing with $\lambda\ge0$, for any $k\in\mathbb{N}$. 
\end{lemma}
\begin{proof}
  The case with $k=2n$ is known and has been proved in
  \cite{Moser}. In the following, we provide a simple proof for the
  general case of $k$, although we are only interested in the
  case $k=1$ later in the paper. Let us define $f_k(\lambda) \defeq
  \E_{\chi^2_k(\lambda)}\bigl[ \log X \bigr]$. The probability
  density function of $\chi^2_k(\lambda)$ is \cite{Muirhead}
  \begin{align}
    p_{\chi^2_{k}(\lambda)} (x) &= \sum_{l=0}^\infty
    \frac{e^{-\frac{\lambda}{2}}\bigl(\frac{\lambda}{2}\bigr)^l }{l!} p_{\chi^2_{k+2l}(0)}(x),
    \quad x\ge0. 
  \end{align}%
  Then it readily follows from the definition of $f_k(\lambda)$ that 
  \begin{align}
    f_k(\lambda) &= \sum_{l=0}^\infty
    \frac{e^{-\frac{\lambda}{2}}\bigl(\frac{\lambda}{2}\bigr)^l}{l!} f_{k+2l}(0).  
  \end{align}%
  To prove that $f_k(\lambda)$ is increasing with $\lambda$, it is
  enough to show that the derivative of $f_k(\lambda)$ with respect to
  $\lambda$ is positive. Indeed,
  \begin{align}
    f_k'(\lambda) &= -\frac{1}{2} f_k(\lambda) + \frac{1}{2} \sum_{l=0}^\infty \frac{e^{-\frac{\lambda}{2}}\bigl(\frac{\lambda}{2}\bigr)^l}{l!} f_{k+2(l+1)}(0) \\
    &= \frac{1}{2} (f_{k+2}(\lambda) - f_{k}(\lambda)) \\
    &= \frac{1}{2} \bigl(\E \bigl[ \log (X+Y) \bigr] - \E \bigl[ \log X\bigr] \bigr) \\
    &> 0, 
  \end{align}%
  where we used the fact that if $X\sim~\chi^2_{k}(\lambda)$ and $Y\sim\chi^2_{2}(0)$, then
  $X+Y\sim\chi^2_{k+2}(\lambda)$. 

\end{proof}

\begin{lemma}[Change of variables~\cite{Cover}]\label{lemma:h-prod}
  Let $\rvVec{Y} = f(\rvVec{X})$ with a bijective map $f:\,\mathbb{R}^m\to \mathbb{R}^m$. Then 
  \begin{align}
    h( \rvVec{Y}) &=  h(\rvVec{X}) + \E \bigl[ \log |\det(\Jm)| \bigr],
  \end{align}%
  where $\Jm \defeq \left[ \frac{\partial Y_k}{\partial X_l}
  \right]_{k,l=1,\ldots,m}$ is the
  Jacobian matrix.
\end{lemma}

\begin{lemma}\label{lemma:cs-vec}
  If each element of the $n$-vector $\pmb{X}$ is circularly symmetric with
  independent phases, and the probability density function~(pdf) of $\pmb{X}$ exists with respect to the Lebesgue
  measure, then  
  \begin{align}
    p_{|X|}(|\pmb{x}|) &= 2\pi \prod_{i=1}^n |x_i|\,p_{X}(\pmb{x}) \\
    p_{|X|^2}(|\pmb{x}|^2) &= \pi \,p_{X}(\pmb{x}). \label{eq:X-X2}
  \end{align}%
  Further, if $h(\pmb{X})>-\infty$, we have 
  \begin{align}
    h(\rvVec{X}) &= h(|\rvVec{X}|) + \sum_{i=1}^n \E \bigl[ \log |X_i| \bigr] + n\log 2\pi \\
    &= h(|\rvVec{X}|^2) + n \log \pi.
  \end{align}
\end{lemma}

\begin{lemma}\label{lemma:h-squared}
  Let $\rvVec{X}\in\mathbb{C}^{n}$ with $h(\rvVec{X})>-\infty$. Then
  \begin{align}
    h(\rvVec{X}) &= h(|\rvVec{X}|^2) + h(\angle_{\rvVec{X}} \cond |\rvVec{X}|) - n.
    \label{eq:h-squared-0}
  \end{align}%
  Let $\rvVec{\Theta}\in [0,2\pi)^{n}$ be independent of $\rvVec{X}$ and 
  $h(\rvVec{\Theta})>-\infty$. Then
  \begin{align}
    | h(e^{j\rvVec{\Theta}} \circ \rvVec{X}) - h(|\rvVec{X}|^2) | \le  \const.  
    \label{eq:h-squared-1}
  \end{align}
\end{lemma}
\begin{proof}
  Applying Lemma~\ref{lemma:h-prod} twice, we readily obtain \eqref{eq:h-squared-0}
  \begin{align}
    h(\rvVec{X}) &= h(|\rvVec{X}|, \angle_{\rvVec{X}} ) + \sum_{k=1}^n \E\log |X_k| \\
    &= h(|\rvVec{X}| ) + h(\angle_{\rvVec{X}} \cond |\rvVec{X}|) + \sum_{k=1}^n \E\log |X_k| \\
    &= h(|\rvVec{X}|^2) + h(\angle_{\rvVec{X}} \cond |\rvVec{X}|) -
    n. 
  \end{align}%
  To prove \eqref{eq:h-squared-1}, we introduce $\rvVec{\Phi}$ that is uniformly distributed in
  $[0,2\pi)^n$ and independent of $\rvVec{X}$ and $\rvVec{\Theta}$, then
  \begin{align}
    h( e^{j\rvVec{\Theta}} \circ \rvVec{X}) &= 
    h(e^{j(\rvVec{\Theta}+\rvVec{\Phi})} \circ \rvVec{X} \cond \rvVec{\Phi}) \\
    &\le h( e^{j\rvVec{\Phi}'} \circ \rvVec{X}) \\
    &= h(|\rvVec{X}|^2) + n\log \pi, \\
    \intertext{where $\rvVec{\Phi}'\defeq (\rvVec{\Theta}+\rvVec{\Phi})_{2\pi}$ is
    uniformly distributed in $[0,2\pi)^n$, and
    from~\eqref{eq:h-squared-0},}
    h(e^{j\rvVec{\Theta}} \circ \rvVec{X}) 
    &= h(|\rvVec{X}|^2) + h( (\angle_{\rvVec{X}} + \rvVec{\Theta})_{2\pi} \cond |\rvVec{X}|) - n \\
    &\ge h(|\rvVec{X}|^2) + h( (\angle_{\rvVec{X}} + \rvVec{\Theta})_{2\pi} \cond |\rvVec{X}|, \angle_{\rvVec{X}}) - n \\
    &= h(|\rvVec{X}|^2) + h(\rvVec{\Theta}) - n.  
  \end{align}%
  Hence, \eqref{eq:h-squared-1} holds with the constant $\const$ corresponding to $\max\left\{ |h(\rvVec{\Theta})-n|, n\log\pi  \right\}$. 
\end{proof}

\begin{lemma}\label{lemma:Elogsin}
  For any ${\Theta}\in [0,2\pi)$ with $h({\Theta})>-\infty$, 
  \begin{align} \label{eq:Elogsin}
    \E \bigl[\log |\sin(\Theta)|\bigr] > \frac{h(\Theta)-\log \bigl( B(\frac{1-\alpha}{2},
    \frac{1}{2})\bigr)}{\alpha},
    \quad \forall\,\alpha\in(0,1),
  \end{align}%
  where $B(x,y)$ is the Beta function. Thus, $\E \bigl[\log |\sin(\Theta)|\bigr]>-\infty$.
  Let $\rvVec{\Theta}\in [0,2\pi)^{n}$. If $h(\rvVec{\Theta})>-\infty$, then
  \begin{align}
    h(\cos(\rvVec{\Theta}))>-\infty. \label{eq:hcos}
  \end{align}%
\end{lemma}
\begin{proof}
  To prove \eqref{eq:Elogsin}, we introduce an auxiliary distribution
  with density $q(\theta) =
  \beta |\sin (\theta)|^{-\alpha}$, $\theta\in[0,2\pi)$, with $\alpha\in(0,1)$ and $\beta \defeq
  \frac{1}{2B(\frac{1-\alpha}{2},\frac{1}{2})}$. Then it
  follows that $h(\Theta) + \E\bigl[ \log(q(\Theta)) \bigr] = - D(p\,\|\, q) \le 0$ where $D(\cdot\|\cdot)$
  is the Kullback-Leibler divergence, which yields \eqref{eq:Elogsin}.
  We proceed to prove \eqref{eq:hcos}, 
  \begin{align}
    h(\cos(\rvVec{\Theta})) &\ge h(\cos(\rvVec{\Theta}) \cond \Omega) \\ 
    &= h(\rvVec{\Theta} \cond \Omega) + \sum_{k=1}^n \E \bigl[ \log |\sin(\Theta_k)|
    \bigr] \\
    &= h(\rvVec{\Theta}) - I(\Omega; \rvVec{\Theta}) + \sum_{k=1}^n \E \bigl[ \log
    |\sin(\Theta_k)| \bigr] \\
    &> -\infty,
  \end{align}%
  where 
  we partition $[0,2\pi)^n$
  in such a way that $\cos(\rvVec{\Theta})$ is a
  bijective function of $\rvVec{\Theta}$ in each partition indexed by $\Omega$; the first
  equality is from Lemma~\ref{lemma:h-prod}; the last inequality is from the boundedness
  of $h(\rvVec{\Theta})$, the fact
  that $\Omega$ only takes a finite number of values, and the
  application of~\eqref{eq:Elogsin}. 
\end{proof}

\begin{lemma} \label{lemma:Elognorm}
  Let $\rvVec{V} \in \mathbb{R}^{m\times1}$ with $h(\rvVec{V}) > -\infty$ and $\E \bigl[
  \|\rvVec{V}\|^2 \bigr] < \infty$. Then
  \begin{align}
    \inf_{\pmb{x}\in \mathbb{R}^m:\, \|\pmb{x}\|=1}  \E
    \bigl[ \log |\rvVec{V}^\T \pmb{x}| \bigr] > -\infty.
  \end{align}%
\end{lemma}
\begin{proof}
  This is a straightforward adaptation of the result in \cite[Lemma 6.7-f]{Moser} for the
  complex case. The real case can be proved by following the same steps. To be self-contained, we provide an
  alternative proof as follows. Define $V_x \defeq |\rvVec{V}^\T
  \pmb{x}|^2$, and one can verify from the assumptions
  that $h(V_x) > -\infty$ and $\E[V_x]\le \infty$. We introduce an
  auxiliary pdf~$q(V_x)$ based on the Gamma distribution defined
  in \eqref{eq:sv-gamma} with some $\alpha\in(0,1)$. Then  we have for any $V_x\in\mathbb{R}^+$, $h(V_x)\le
  \E\bigl[ -\log q(V_x) \bigr] = (1-\alpha) \E\bigl[\log V_x \bigr] + \mu \E[ V_x ] + \const$ which yields 
  $\E\bigl[\log V_x \bigr] \ge (1-\alpha)^{-1}  (h(V_x) - \mu \E[ V_x ] - \const) > -\infty.$ 
\end{proof}

\begin{lemma}
  \label{lemma:ejt}
  Let $\Theta\in[0,2\pi)$ with $h(\Theta) > -\infty$  
  and be independent of some $Z\sim\mathcal{CN}(0,1)$, then for any
  given $\beta\in\mathbb{C}$,
  \begin{align}
    \left| h(\beta e^{j\Theta} + Z) - \log^+\! |\beta|
    \right| &\le \const, \label{eq:Theta} 
    \\ \text{and} \quad \bigl| h( |\beta + Z | ) \bigr| &\le \const'.
    \label{eq:beta} 
  \end{align}%
\end{lemma}
\begin{proof}
  First  we prove \eqref{eq:Theta}. 
  When $|\beta| \le 1$, we have $\log^+\! |\beta|=0$. It
  follows that $h(\beta e^{j\Theta} + Z) \ge h(\beta
  e^{j\Theta} + Z \cond \Theta) = \log (\pi e)$ and
  $h(\beta e^{j\Theta} + Z) \le \log (\pi e
  (\mathsf{Var}(\beta e^{j\Theta} + Z))) \le \log
  (2\pi e)$, which proves \eqref{eq:Theta} for $|\beta|\le1$. 
  Next, we assume that $|\beta|>1$. It is without loss of generality to
  consider $\beta\in\mathbb{R}^+$.
  Let $Z_{\text{R}}$ and $Z_{\text{I}}$ be the
  real and imaginary parts of $Z$, respectively. Then
  \begin{align}
    h(\beta e^{j\Theta} + Z) &=  
    h(\beta \cos(\Theta) + Z_{\text{R}}) \nonumber \\
    &\quad + h(\beta \sin(\Theta) + Z_{\text{I}} \cond 
    \beta \cos(\Theta) + Z_{\text{R}}) \\
    &\ge h(\beta \cos(\Theta)) + h(Z_{\text{I}}) \\
    &= \log \beta + h(\cos(\Theta)) + \frac{1}{2}\log (\pi e). 
    \label{eq:tmp110}
  \end{align}%
  Since $\beta e^{j\Theta} + Z = e^{j\Theta} (\beta + \tilde{Z})$ where
  $\tilde{Z}\defeq Z e^{-j\Theta}\sim\mathcal{CN}(0,1)$ is
  independent of $\Theta$, we can apply \eqref{eq:h-squared-1} from Lemma~\ref{lemma:h-squared}, 
  \begin{align}
    \MoveEqLeft[0]{h(\beta e^{j\Theta} + Z)} \nonumber \\
    &\le 
    h(\beta^2  + |Z|^2 + 2 \beta |Z| \cos(\Theta-\angle_Z)) + \const \\ 
    &= h(|Z|^2 + 2 \beta |Z| \cos(\Theta-\angle_Z)) + \const  \\ 
    &\le \frac{1}{2} \log \left( 2 \pi e \mathsf{Var}(|Z|^2 + 2 \beta |Z| \cos(\Theta-\angle_Z)) \right)
    + \const \\
    &\le \log \beta + \frac{1}{2} \log \left( 2 \pi e\bigl( \mathsf{Var}(|Z|^2) + 
    \mathsf{Var}(2 |Z|) \bigr) \right) +
    \const, 
    \label{eq:tmp111}
  \end{align}%
  where we use the condition $\beta>1$. 
  The lower bound \eqref{eq:tmp110} and upper bound
  \eqref{eq:tmp111} complete the proof of
  \eqref{eq:Theta} for $|\beta|>1$. To prove \eqref{eq:beta}, we
  introduce some $\Theta$ uniformly distributed in
  $[0,2\pi)$, then $ h( |\beta + Z | ) = h(|\beta
  e^{j\Theta} + Z |) = h(\beta e^{j\Theta} + Z) -
  \E \bigl[ \log(|\beta + Z|) \bigr] - \log2\pi$. It can be shown that
  $\bigl| \E \bigl[\log(|\beta + Z|)\bigr] - \log^+\!|\beta| \bigr|
  \le \const$. To see this, we write $\E \bigl[ \log(|\beta + Z|) \bigr] =
  \frac{1}{2} \E \bigl[ \log(|Z|^2+|\beta|^2+2|\beta
  Z|\cos(\angle_{\beta^*\! Z})) \bigr]$ where $\angle_{\beta^*\!
  Z}$ is uniformly distributed in $[0,2\pi)$ and
  independent of the other variables. Taking expectation over
  $\angle_{\beta^*\! Z}$, we
  obtain $\E \bigl[ \log(|\beta + Z|) \bigr] = \frac{1}{2} \E \bigl[
  \log(|Z|^2+|\beta|^2) \bigr] + \const$, since $\frac{1}{2\pi}\int_0^{2\pi}
  \log(a+b\cos \theta) \text{d} \theta = \log \frac{a+\sqrt{a^2-b^2}}{2}$, for all
  $a\ge b >0$. Then, applying
  Jensen's inequality with expectation over $Z$, we have $\E \bigl[ \log(|\beta +
  Z|) \bigr]  \le \frac{1}{2}\log(1+|\beta|^2) + \const \le
  \log^+\!|\beta| + \const'$. Using the monotonicity of
  the logarithmic function, we also have $\E \bigl[ \log(|\beta +
  Z|) \bigr] = \frac{1}{2} \E \bigl[
  \log(|Z|^2+|\beta|^2) \bigr] + \const \ge \max\{\log |\beta|, \E \bigl[\log |Z| \bigr]\} + \const \ge
  \log^+\!|\beta| + \const'$. Finally, since both
  $h(\beta e^{j\Theta} + Z)$ and $\E \bigl[ \log(|\beta + Z|) \bigr]$
  are ``close'' to $\log^+\! |\beta|$, they are ``close''
  to each other due to the triangle inequality. This
  completes the proof of \eqref{eq:beta}. 
\end{proof}

\begin{lemma} \label{lemma:log+}
  For any $p,X>0$, we have 
  \begin{align}
    |\log^+\! (pX) - \log^+\! X| &\le \lvert \log p \rvert, \quad \text{and}\label{eq:logAX}\\
     \E\bigl[ \log^+\! X \bigr] &\le p^{-1} \log^+\! \bigl(\E[ X^p ]\bigr) + p^{-1}. \label{eq:Jensen}
  \end{align}%
\end{lemma}
\begin{proof}
  To show \eqref{eq:logAX}, it is enough to verify that 
  $\log^+\! (pX) - \log^+\! X \le \log p$ when $p\ge 1$ and 
  $\log^+\! X - \log^+\! (pX) \le - \log p$ when $p < 1$, which completes the proof. The inequality
  \eqref{eq:Jensen} is based on Jensen's inequality. Specifically, we have $\E\bigl[ \log^+\! X \bigr]
  = p^{-1} \E\bigl[ \log^+\! X^p \bigr] \le p^{-1} \E\bigl[ \log( 1 + X^p )\bigr] \le p^{-1} \log
  \bigl(1+\E[ X^p ]\bigr) \le p^{-1} \log^+\! \bigl(\E[ X^p ]\bigr) + p^{-1}$.
\end{proof}


\section{Capacity Upper Bound for Model~A}
\label{sec:upperbound}

The capacity $C(P)$ in \eqref{eq:capa} of a stationary
and ergodic channel is upper-bounded by the capacity of
the corresponding memoryless channel up to a constant
term. Following the footsteps of
\cite{Moser,Lapidoth-ITW02}, we have 
\begin{align}
  \frac{1}{N} I(\rvVec{X}^N; \rvVec{Y}^N) &= \frac{1}{N}
  \sum_{k=1}^N
  I(\rvVec{X}^N; \rvVec{Y}_k \cond \rvVec{Y}^{k-1}) \\
  &\le \frac{1}{N} \sum_{k=1}^N I(\rvVec{X}_k; \rvVec{Y}_k) + I(\rvMat{\Theta}_N; \rvMat{\Theta}^{N-1}) \\
  &\le \sup I(\rvVec{X}; \rvVec{Y}) + \const,
\end{align}%
where $\sup I(\rvVec{X}; \rvVec{Y})$ is the capacity
of a memoryless phase noise channel with the same temporal
marginal distribution as the original channel, and
the supremum is over all input distributions such that
 $\E \left[ \|\rvVec{X}\|^2 \right] \le P$; using the fact that $I(\rvMat{\Theta}_N;
 \rvMat{\Theta}^{N-1}) = h(\rvMat{\Theta}_N) - h(\rvMat{\Theta}_N \cond \rvMat{\Theta}^{N-1}) \le
 \log (2\pi) - r_{\Theta}$ where $r_{\Theta}$ is the differential entropy rate of the phase noise
 process, we can set $\const = \log (2\pi) - r_{\Theta}$. 
 Since we are mainly interested in the multiplexing gain, the constant
 $\const$ does not matter, and it is thus without loss of optimality to consider the memoryless case in this section. 

 The main ingredients of the proof are the genie-aided bound and the duality upper
 bound. In the following, we detail the five steps that lead to 
 Theorem~\ref{thm:A}. 

\subsection{Step~1: Genie-aided bound}
\label{sec:GAB}

Let us define the auxiliary random variable $U$ as the index of the strongest input
entry, i.e.,\footnote{When there are more than one such elements, we pick an arbitrary one.} 
\begin{align}
  U &\defeq \arg\max_{1\le i\le \nt} |X_i|.
\end{align}%
Thus, we use $X_U$ to denote the element in $\rvVec{X}$ with the largest magnitude. 
It is obvious that $U \leftrightarrow \rvVec{X} \leftrightarrow \rvVec{Y}$
form a Markov chain, and that $U$ does not contain more than $\log \nt$ bits. Assuming that a genie provides $U$ to the receiver, we obtain the
following upper bound
\begin{align}
  I(\rvVec{X}; \rvVec{Y}) &\le I(\rvVec{X}; \rvVec{Y}, U) \\
  &= I(\rvVec{X}; \rvVec{Y} \cond U) + I(U; \rvVec{X}) \\
  &\le I(\rvVec{X}; \rvVec{Y} \cond U) + H(U) \\
  &\le I(\rvVec{X}; \rvVec{Y} \cond U) + \log \nt. \label{eq:tmp8920}
\end{align}%

\subsection{Step~2: Canonical form}

\begin{definition}[Canonical channel]\label{lemma:model12}
We define the canonical form $u$, $u=1,\ldots,\nt$, of the channel $\rvMat{H}$ as
\begin{align}
  \rvMat{G}_{(u)} \defeq \underbrace{\diag\left( h_{1,u}^{-1},\ldots,h_{\nr,u}^{-1}
  \right)}_{\Am_u} \rvMat{H}. \label{eq:canonical}
\end{align}
\end{definition}
Note that the elements in the $u$\,th column of $\rvMat{G}_{(u)}$ has normalized magnitudes.
Now, with the information $U$ from the genie, the receiver can convert
the original channel into one of the canonical forms, namely, the
form~$U$. 
\begin{align}
  {I(\rvVec{X}; \rvMat{H} \rvVec{X} + \rvVec{Z} \cond U)} 
  &= I(\rvVec{X}; \Am_U \rvMat{H} \rvVec{X} + \Am_U \rvVec{Z} \cond U) \\
  &\le I(\rvVec{X}; \Am_U \rvMat{H} \rvVec{X} + a \rvVec{Z} \cond U) \label{eq:tmp322}\\
  &= I(a^{-1} \rvVec{X}; a^{-1} \rvMat{G}_{(U)} \rvVec{X} + \rvVec{Z} \cond U) \\
  &= I( \tilde{\rvVec{X}}; \rvMat{G}_{(U)} \tilde{\rvVec{X}} + \rvVec{Z} \cond U), \label{eq:tmp820}
\end{align}
where $a \defeq \min_{k,u} |h_{k,u}^{-1}|$; \eqref{eq:tmp322} is due to the fact that reducing the additive noise
increases the mutual information; we define 
\begin{align}
  \tilde{\rvVec{X}} &\defeq a^{-1} \rvVec{X},
  \label{eq:Xtilde}\\
\intertext{and accordingly,}
{\rvVec{W}} &\defeq \rvMat{G}_{(u)}
\tilde{\rvVec{X}} + \rvVec{Z}.
\end{align}%
In the following, we focus on upper-bounding the mutual information 
$I(\tilde{\rvVec{X}}; {\rvVec{W}} \cond U)$. Note that 
\begin{align}
I( \tilde{\rvVec{X}}; {\rvVec{W}} \cond U ) &= h(
{\rvVec{W}} \cond U) - h( {\rvVec{W}} \cond
\tilde{\rvVec{X}}, U)\\
&= h( {\rvVec{W}} \cond U) - h( {\rvVec{W}} \cond \tilde{\rvVec{X}}),
\label{eq:I-XW}
\end{align}%
where the last equality comes from the fact that $U$ is a function of
$\rvVec{X}$ and thus a function of $\tilde{\rvVec{X}}$, since
$\tilde{\rvVec{X}}$ is simply a scaled version of $\rvVec{X}$.
Therefore, it is enough to lower-bound $h( {\rvVec{W}} \cond
\tilde{\rvVec{X}})$ and upper-bound $h( {\rvVec{W}} \cond U)$
separately.

\subsection{Step~3: Lower bound on $h( {\rvVec{W}} \cond \tilde{\rvVec{X}})$}

\begin{lemma}\label{lemma:LB-A}
  For model~A, we have 
  \begin{align}
    h({\rvVec{W}} \cond \tilde{\rvVec{X}})&\ge \nr
    \,\E\bigl[ \log^+\! |\tilde{X}_U| \bigr]\! + \nr \,\E \bigl[ \log^+ \!
    |\tilde{X}_V| \bigr]\!  + \constH, \label{eq:LB-A}
  \end{align}%
  where $\tilde{X}_U$ and $\tilde{X}_V$ have the largest and second largest magnitudes in $\tilde{\rvVec{X}}$, respectively.
\end{lemma}
 \begin{proof}
   See Appendix~\ref{app:LB}. 
 \end{proof}
It is worth mentioning that the above bound depends not only on the strongest but also on the second strongest input of the channel.

\subsection{Step~4: Upper bound on $h( {\rvVec{W}} \cond U)$}

Upper-bounding $h({\rvVec{W}} \cond U)$ by a non-trivial but
tractable function of the input distribution is hard in
general. A viable way for that purpose is through an
auxiliary distribution, also called the duality approach. The duality upper bound was first proposed
in \cite{Topsoe} and \cite{Kemperman} for discrete channels and then derived for arbitrary channels
in~\cite{Moser}. Namely, for
any\footnote{Formally, we
should state that the probability measure $Q$ corresponding to the density~$q(\pmb{w})$ is such that
$P(\cdot\cond U=u)$ is absolutely continuous with respect to $Q$. Throughout the paper, for brevity, we
implicitly make the assumption to avoid such formalities.} pdf
$q(\pmb{w})$, we have 
\begin{align}
  h({\rvVec{W}} \cond U) &= \E \bigl[ - \log
  p({\rvVec{W}} \cond U) \bigr]
  \\
  &= \E \bigl[ - \log
  q({\rvVec{W}}) \bigr] - \E_U \bigl[ D( p_{{\rvVec{W}}|U=u} \,\|\, q )\bigr] \\
  &\le \E\bigl[ - \log q({\rvVec{W}}) \bigr] \label{eq:tmp7299}
\end{align}%
due to the non-negativity of the Kullback-Leibler divergence 
$D( p_{{\rvVec{W}}|U=u}\, \|\, q )$. 
Hence, the key is to choose a proper
auxiliary pdf $q(\pmb{w})$ in order to obtain a tight upper
bound on the capacity of our channel. The commonly
used auxiliary distributions for MIMO channels are mostly related to
the class of isotropic
distributions~\cite{Moser,Lapidoth-ITW02,Durisi-capa}. Unfortunately,
the isotropic distributions are not suitable in our
case. To see this, let us assume that an isotropic output
${\rvVec{W}}$ was
indeed close to optimal. On the one hand, the pdf of an isotropic
output ${\rvVec{W}}$ would only
depend on the norm $\|{\rvVec{W}}\|$ which would be dominated by the largest input entry
$X_U$ at high SNR. Therefore, the value of $\E \bigl[ - \log
q({\rvVec{W}}) \bigr]$ would be insensitive to the
number of active input entries. On the other hand, the lower bound on the conditional entropy
$h({\rvVec{W}} \cond \tilde{\rvVec{X}})$ is increasing with \emph{both} of the largest input
entries $X_U$ and $X_V$, according to \eqref{eq:LB-A}. Therefore,
with an isotropic distribution $q(\pmb{w})$, the
capacity upper bound $\E\bigl[ - \log q({\rvVec{W}}) \bigr] - h({\rvVec{W}} \cond
\tilde{\rvVec{X}})$ would become larger when the second strongest
input went to zero, i.e., only one
transmit antenna was active. But this is in contradiction with the isotropic assumption, since if
only one transmit antenna was active, then the output entries would
be highly correlated and the
output distribution would be \emph{far} from being isotropic. 

In light of the above discussion, we are led to think
that a \emph{good} choice of $q(\wv)$ should reflect not only
the strongest input entry, but also the weaker ones. We adopt the following pdf built
from the multivariate Gamma distribution in Definition~\ref{def:mv-Gamma},
  \begin{align} 
    q(\wv) &= \frac{g_{\pmb{\alpha}}}{\nr!}
    |\hat{w}_1|^{2(\alpha_1-1)} \prod_{i=2}^{\nr} \left( |\hat{w}_i|^2 -
    |\hat{w}_{i-1}|^2
    \right)^{\alpha_i-1} \nonumber \\
    &\qquad \cdot \exp(-\mu |\hat{w}_{\nr}|^2)
    \mu^{\alpha_1+\cdots+\alpha_{\nr}}, \quad \pmb{w} \in
    \mathbb{C}^{\nr}, \label{eq:bi-gamma}
  \end{align}%
  where $\hat{w}_1, \ldots,
  \hat{w}_{\nr}$ are the ordered version of $w_i$'s with
  increasing magnitudes. Essentially, we let each $W_i$ be circularly
  symmetric and let the ordered version of $(|W_1|^2,
  \ldots, |W_{\nr}|^2)$ follow the multivariate Gamma distribution
  defined in Definition~\ref{def:mv-Gamma}. Applying \eqref{eq:X-X2}
  in Lemma~\ref{lemma:cs-vec} and the order statistics~(whence
  the term $\nr!$)~\cite{Muirhead}, we can obtain the pdf of
  $\rvVec{W}$ as written in \eqref{eq:bi-gamma}. Remarkably,
  the differences
  between $|W_i|^2$ and $|W_j|^2$, $i\ne j$, are introduced into the
  upper bound, which is crucial for bringing in the impact of
  individual input entries $\tilde{X}_i$'s other than the strongest
  entry as will be shown in the following.  

  \begin{lemma}
    \label{lemma:logq}
    By choosing $0<\alpha_i<1$, $i=1,\ldots,\nr$, and $\mu =
    \min\{P^{-1},1\}$, we
    have for model A,  
    \begin{align}
      \MoveEqLeft[0.5]{\E\bigl[ - \log q({\rvVec{W}})
      \bigr]} \nonumber \\
      &\le \sum_{i=1}^{\nr} \alpha_i \log^+\! P + \left((1-\alpha_1) + \sum_{i=1}^{\nr}
  (1-\alpha_i)\right)   \E \bigl[ \log^+\! |\tilde{X}_U| \bigr] \nonumber \\
  &\qquad + 
  \sum_{i=2}^{\nr} (1-\alpha_i) \E \bigl[ \log^+\! |\tilde{X}_V| \bigr]+ \constH, \label{eq:Elogq}
    \end{align}%
  where $\tilde{X}_U$ and $\tilde{X}_V$ are the strongest and second
  strongest elements in $\tilde{\rvVec{X}}$, respectively.
 \end{lemma}
 \begin{proof}
   The calculation is straightforward from the pdf
   \eqref{eq:bi-gamma}, details are provided in Appendix~\ref{app:logq}.
 \end{proof}

\subsection{Step~5: Upper bound for model A}
Combining \eqref{eq:I-XW}, \eqref{eq:LB-A}, \eqref{eq:tmp7299}, and \eqref{eq:Elogq}
from the previous steps, we have 
\begin{align}
  \MoveEqLeft[0]{I( \tilde{\rvVec{X}}; {\rvVec{W}} \cond U)} \nonumber \\
  &\le \sum_{i=1}^{\nr} \alpha_i \log^+\! P + \left(1-2\alpha_1 - \sum_{i=2}^{\nr}
  \alpha_i\right) \E \bigl[ \log^+\! |\tilde{X}_U| \bigr] \nonumber \\
  &\qquad + 
  \left(\sum_{i=2}^{\nr} (1-\alpha_i) - {\nr}\right) \E \bigl[ \log^+\!
  |\tilde{X}_V| \bigr]  + \constH \label{eq:tmp911} \\
  &\le \sum_{i=1}^{\nr} \alpha_i \log^+\! P +  \left(1-2\alpha_1 - \sum_{i=2}^{\nr}
      \alpha_i\right)  \E \bigl[ \log^+\! |\tilde{X}_U| \bigr] + \constH \label{eq:tmp912}
 \\
      &\le \left( \sum_{i=1}^{\nr} \alpha_i  +  \frac{1}{2} \right) \log^+\!
      P + \constH',
    \end{align}
   where the inequality \eqref{eq:tmp912} comes from
    removing the negative term in \eqref{eq:tmp911}; to obtain the last inequality, 
    we apply \eqref{eq:Jensen} in Lemma~\ref{lemma:log+} with $p=2$ and the power constraint $\E \bigl[ |\tilde{X}_U|^2 \bigr] \le a^2 \E \bigl[
    \|\rvVec{X}\|^2 \bigr] \le a^2 P$.  

Finally, we conclude from \eqref{eq:tmp8920} and \eqref{eq:tmp820} that, for model~A,  
\begin{align}
  I(\rvVec{X}; \rvVec{Y}) &\le I(\rvVec{X}; \rvVec{Y}\cond U) +
  \const \\
  &\le I( \tilde{\rvVec{X}}; {\rvVec{W}} \cond U) + \const \\
  &\le \left( \sum_{i=1}^{\nr} \alpha_i  +  \frac{1}{2} \right) \log^+\! P + \constH 
\end{align}%
which implies that the multiplexing gain is upper-bounded by 
\begin{align}
  r_\text{A} &\le \sum_{i=1}^{\nr} \alpha_i  +  \frac{1}{2}, \quad
  \forall\, \pmb{\alpha}\in(0,1)^{\nr}.
\end{align}%
By taking the infimum over $\pmb{\alpha}$, we have $r_{\text{A}} \le \frac{1}{2}$.

\section{Capacity Upper Bound for Model~B}
\label{sec:upperbound-B}

In this section, we derive upper bounds for the three cases of
model~B, where the phase noises are on the transmitter and receiver
sides of the channel. As in the previous section, it is enough to
consider the memoryless case for our purpose.

\subsection{Case~B1: Transmit and receive phase noises}

Note that the multiplexing gain of this case is upper-bounded by
that of case~B2 and case~B3, since we can enhance the channel
by providing the information on the transmit or receive phase noises
to both the transmit and receiver. In other words, the upper bound
$\min\{\frac{\nr}{2}, \frac{\nt}{2}, \nt-\frac{1}{2}\} = \min\{\frac{\nr}{2},
\frac{\nt}{2}\}$ is still valid for this case. In the following, we
show that we can tighten the upper bound $\frac{\nr}{2}$ to
$\frac{(\nr-2)^+ + 1}{2}$ with the duality upper bound using the
multi-variate Gamma distribution. The proof is in
the same vein as the proof for model~A. Specifically, the first four
steps are exactly the same as for model~A, except for Step~3~in which
the conditional entropy has a different lower bound, as shown below. 
\begin{lemma}\label{lemma:LB-B}
  For model~B1, we have
  \begin{align}
   h({\rvVec{W}} \cond \tilde{\rvVec{X}})
    &\ge \nr \,\E \bigl[ \log^+\! |\tilde{X}_U| \bigr] + \E \bigl[ \log^+ \!
    |\tilde{X}_V| \bigr] + \constH, \label{eq:LB-B}
  \end{align}%
  where $\tilde{X}_U$ and $\tilde{X}_V$ have the largest and second largest magnitudes in $\tilde{\rvVec{X}}$, respectively.
\end{lemma}
\begin{proof}
  See Appendix~\ref{app:LB}. 
 \end{proof}

Applying \eqref{eq:I-XW},
 \eqref{eq:tmp7299}, \eqref{eq:Elogq}, and \eqref{eq:LB-B}, we have
\begin{align}
  \MoveEqLeft[2]{I( \tilde{\rvVec{X}}; {\rvVec{W}} \cond U)} \nonumber \\
  &\le \sum_{i=1}^{\nr} \alpha_i \log^+\! P + \left(1-2\alpha_1 - \sum_{i=2}^{\nr}
  \alpha_i\right) \E \bigl[ \log^+\! |\tilde{X}_U| \bigr] \nonumber \\
  &\qquad + 
  \left(\sum_{i=2}^{\nr} (1-\alpha_i) - 1\right) \E \bigl[ \log^+\!
  |\tilde{X}_V| \bigr] + \constH \\
  &\le \sum_{i=1}^{\nr} \alpha_i \log^+\! P + \E \bigl[ \log^+\! |\tilde{X}_U|\bigr]
  \nonumber \\ &\qquad + \left({\nr}-2\right)^+ \E \bigl[ \log^+\! |\tilde{X}_V| \bigr] + \constH \\
  &\le \sum_{i=1}^{\nr} \alpha_i \log^+\! P + \left( ({\nr}-2)^++1\right)
  \E \bigl[ \log^+\! |\tilde{X}_U| \bigr] + \constH  \nonumber \\
  &\le \Bigl( \frac{({\nr}-2)^++1}{2} + \sum_{i=1}^{\nr} \alpha_i \Bigr) \log^+\! P
  + \constH'. \label{eq:tmp4443}
\end{align}

Therefore, we conclude from \eqref{eq:tmp8920}, \eqref{eq:tmp820},
and \eqref{eq:tmp4443} that, for model~B1,
\begin{align}
  I(\rvVec{X}; \rvVec{Y}) &\le \left( \sum_{i=1}^{\nr} \alpha_i  +
  \frac{({\nr}-2)^++1}{2} \right) \log^+\! P + \constH' 
\end{align}%
which implies that the multiplexing gain is upper-bounded by 
\begin{align}
  r_\text{B} &\le \sum_{i=1}^{\nr} \alpha_i  +  \frac{ ({\nr}-2)^++1 }{2}, \quad \forall \,
  \pmb{\alpha}\in(0,1)^{\nr}.
\end{align}%
Taking the infimum over $\pmb{\alpha}$, we have 
$r_\text{B} \le  \frac{({\nr}-2)^++1}{2}.$

\subsection{Case~B2: Transmit phase noise}

In this case, the received signal is $\rvVec{Y} =
\pmb{H} (e^{j\rvVec{\Theta}_{\text{T}}} \circ
\rvVec{X}) + \rvVec{Z}$. The channel is characterized by the random matrix
$\rvMat{H} = \pmb{H} \diag\{e^{j\rvVec{\Theta}_\text{T}}\}$. We shall show that the upper
bound is $\min\bigl\{\frac{\nt}{2}, \frac{\nr}{2}\bigr\}$.
First, with more receive antennas than transmit antennas, i.e., when
$\nr\ge \nt$, we can inverse the channel without losing information,
\begin{align}
  I(\rvVec{X}; \rvVec{Y})
  &= I(\rvVec{X}; e^{j\rvVec{\Theta}_{\text{T}}}\circ \rvVec{X} +
  \pmb{H}^\dag \rvVec{Z} ) \\
  &\le I(\rvVec{X}; e^{j\rvVec{\Theta}_{\text{T}}}\circ \rvVec{X} +
  \tilde{\rvVec{Z}}),
  \label{eq:tmp4478} 
\end{align}%
where $\tilde{\rvVec{Z}} \sim \mathcal{CN}(0, \sigma^2_{\min}(\pmb{H}^\dag)
\Id_{\nt})$, with
$\sigma_{\min}(\pmb{H}^\dag)>0$ being the minimum singular value of
$\pmb{H}^\dag$. Note that
\eqref{eq:tmp4478} is maximized when $\rvVec{X}$ is circularly
symmetric with $\nt$ independent phases. To see this, we introduce a
vector of independent and identically distributed~(i.i.d.)~phases $\rvVec{\Phi}$ uniformly distributed in
$[0,2\pi)^{\nt}$ and show that, for any $\rvVec{X}$,
\begin{align}
  \MoveEqLeft[0.5]{I(e^{j\rvVec{\Phi}}\circ\rvVec{X};
  e^{j(\rvVec{\Theta}_{\text{T}}+\rvVec{\Phi})} \circ \rvVec{X} +
  \tilde{\rvVec{Z}}  ) } \nonumber \\
  &= h(e^{j(\rvVec{\Theta}_{\text{T}}+\rvVec{\Phi})} \circ
  \rvVec{X} + \tilde{\rvVec{Z}}) - 
  h(e^{j(\rvVec{\Theta}_{\text{T}}+\rvVec{\Phi})}\circ \rvVec{X}
  + \tilde{\rvVec{Z}}  \cond e^{j\rvVec{\Phi}}\circ\rvVec{X}) \\
   &\ge h(e^{j(\rvVec{\Theta}_{\text{T}}+\rvVec{\Phi})}\circ
  \rvVec{X} + \tilde{\rvVec{Z}}  \cond \rvVec{\Phi}) -
h(e^{j(\rvVec{\Theta}_{\text{T}}+\rvVec{\Phi})}\circ \rvVec{X}
  + \tilde{\rvVec{Z}}  \cond \rvVec{X}, \rvVec{\Phi}) \\
   &= h(e^{j\rvVec{\Theta}_{\text{T}}}\circ
  \rvVec{X} + \tilde{\rvVec{Z}}' ) -
h(e^{j\rvVec{\Theta}_{\text{T}}}\circ \rvVec{X}
  + \tilde{\rvVec{Z}}'  \cond \rvVec{X}) \\
  &= I(\rvVec{X}; e^{j\rvVec{\Theta}_{\text{T}}}\circ \rvVec{X} +
  \tilde{\rvVec{Z}}),
\end{align}%
where we use the fact that $\tilde{\rvVec{Z}}$ is circularly
symmetric, and has the same distribution as $\tilde{\rvVec{Z}}' \defeq
e^{-j\rvVec{\Phi}} \tilde{\rvVec{Z}}$. 
Therefore, to derive an upper bound, it is without loss of optimality to assuming that
$\rvVec{X}$ is circularly symmetric with $m$ independent phases. With
this assumption, we have 
\begin{align}
  \MoveEqLeft[0]{I(\rvVec{X}; e^{j\rvVec{\Theta}_{\text{T}}}\circ \rvVec{X} +
  \tilde{\rvVec{Z}})} \nonumber\\
  &= 
  I(|\rvVec{X}|; e^{j\rvVec{\Theta}_{\text{T}}}\circ \rvVec{X} +
  \tilde{\rvVec{Z}} ) + I(\angle_{\rvVec{X}};
  e^{j\rvVec{\Theta}_{\text{T}}}\circ \rvVec{X} + \tilde{\rvVec{Z}}  \cond |\rvVec{X}|) \\
  &\le I(|\rvVec{X}|;
  e^{j(\rvVec{\Theta}_{\text{T}}+\angle_{\rvVec{X}})}\circ
  |\rvVec{X}| + \tilde{\rvVec{Z}} ) +
  I(\angle_{\rvVec{X}}; e^{j\rvVec{\Theta}_{\text{T}}}\circ \rvVec{X}  \cond |\rvVec{X}|) \\
  &\le I(|\rvVec{X}|;
  e^{j(\rvVec{\Theta}_{\text{T}}+\angle_{\rvVec{X}})}\circ |\rvVec{X}|
  + \tilde{\rvVec{Z}}  \cond
  \rvVec{\Theta}_{\text{T}}+\angle_{\rvVec{X}}) \nonumber \\
  &\qquad + 
  I(\angle_{\rvVec{X}};
  (\rvVec{\Theta}_{\text{T}}+\angle_{\rvVec{X}})_{2\pi}) \\
  &= I(|\rvVec{X}|;  |\rvVec{X}| + \Re\bigl\{\tilde{\rvVec{Z}}''\bigr\}) +
  h( (\rvVec{\Theta}_{\text{T}}+\angle_{\rvVec{X}})_{2\pi})
  - h(\rvVec{\Theta}_{\text{T}}) \\
  &\le \frac{\nt}{2} \log^+\! P + \constH + \log2\pi -
  h(\rvVec{\Theta}_{\text{T}} ), \label{eq:tmp0930} 
\end{align}%
where the second inequality is obtain by providing
$\rvVec{\Theta}_{\text{T}}+\angle_{\rvVec{X}}$ to the output and the
independence between $\rvVec{\Theta}_{\text{T}}+\angle_{\rvVec{X}}$
and $|\rvVec{X}|$; the last inequality is from the capacity upper
bound for a real-value Gaussian channel, and the fact that
$(\rvVec{\Theta}_{\text{T}}+\angle_{\rvVec{X}})_{2\pi}$ is uniformly
distributed in $[0,2\pi)$; we define $\tilde{\rvVec{Z}}''\defeq
e^{-j(\rvVec{\Theta}_{\text{T}}+\angle_{\rvVec{X}})} \circ
\rvVec{Z} $. From \eqref{eq:tmp0930}, we get the
upper bound $\frac{\nt}{2}$ of the pre-log. 
 
In the following, we assume $\nr\le \nt$, and follow closely to the proof for model~A in
Section~\ref{sec:GAB}. We first apply a genie-aided bound, by providing the set of indices of the
$\nr$ strongest inputs to the receiver. This information, also denoted by
$U$, does not take more than $\log {\nt\choose{\nr}}$ bits. Then 
we also associate with each $U$ a canonical form $\rvMat{G}_{(U)} =
\pmb{H}_{{U}}^{-1}  \rvMat{H}$ 
where $\pmb{H}_{{U}}$ is the submatrix of $\Hm$ with the columns
corresponding to the $\nr$ strongest entries, while $\Hm_{\bar{U}}$
corresponds to the rest of the columns. It follows that
$\pmb{H}_{{U}}^{-1} \rvVec{Y} = \rvMat{G}_{(U)} \rvVec{X} + \pmb{H}_{{U}}^{-1}
\rvVec{Z}$, 
and
\begin{align}
  {I(\rvVec{X}; \rvMat{H} \rvVec{X} + \rvVec{Z} \cond U)} 
  &= I(\rvVec{X};\rvMat{G}_{(U)} \rvVec{X} + \pmb{H}_{{U}}^{-1} \rvVec{Z}  \cond U) \\
  &\le I(\rvVec{X};\rvMat{G}_{(U)} \rvVec{X} + a \rvVec{Z}
  \cond U) 
  \\
  &= I(a^{-1} \rvVec{X}; a^{-1} \rvMat{G}_{(U)} \rvVec{X} + \rvVec{Z} \cond U) \\
  &= I( \tilde{\rvVec{X}}; \rvMat{G}_{(U)}
  \tilde{\rvVec{X}} + \rvVec{Z} \cond U), 
\end{align}
where $a \defeq (\sigma_{\max}(\pmb{H}))^{-1}$; we define 
  $\tilde{\rvVec{X}} \defeq a^{-1} \rvVec{X}$ and  
accordingly,
\begin{align}
\rvVec{W} \defeq \rvMat{G}_{(U)}
\tilde{\rvVec{X}} + \rvVec{Z} =
e^{j\rvVec{\Theta}_{\text{T},{U}}}
\!\! \circ \tilde{\rvVec{X}}_{{U}} +
\pmb{H}_{{U}}^{-1} \pmb{H}_{\bar{U}}
(e^{j\rvVec{\Theta}_{\text{T},\bar{U}}} \circ
\tilde{\rvVec{X}}_{\bar{U}}) + \rvVec{Z}.
\end{align}%
The next step is to derive a lower bound on $h(\rvVec{W}\cond
\tilde{\rvVec{X}})$,
\begin{align}
  \MoveEqLeft{
  h(\rvVec{W}\cond \tilde{\rvVec{X}})} \nonumber \\
  &\ge  
   h(\rvVec{W}\cond \tilde{\rvVec{X}}, \rvVec{\Theta}_{\text{T},\bar{U}}) \\
   &= h(e^{j\rvVec{\Theta}_{\text{T},{U}}}
\circ \tilde{\rvVec{X}}_{{U}} + \rvVec{Z}\cond \tilde{\rvVec{X}},
\rvVec{\Theta}_{\text{T},\bar{U}}) \\
&\ge \sum_{k=1}^{\nr}  h(e^{j{\Theta}_{\text{T},k}}
 \tilde{{X}}_{k} + {Z}_k\cond \tilde{\rvVec{X}},
 \rvVec{\Theta}_{\text{T},\bar{U}},\Theta_{\text{T}}^{k-1},Z^{k-1})
 \\
 &\ge \sum_{k=1}^{\nr} \E \bigl[ \log^+\! |\tilde{X}_k| \bigr] + \const, \label{eq:tmp8222}
\end{align}%
where we assume that $U=\{1,\ldots,\nr\}$
for notational convenience, and the last inequality is from
Lemma~\ref{lemma:ejt}. 

Finally, we derive an upper bound on
$h(\rvVec{W}\cond U)$ via duality using the 
following auxiliary distribution on the output~$\rvVec{W}$, 
  \begin{align} 
    q(\wv) = {g_{\pmb{\alpha}}}\,\mu^{\alpha_1+\cdots+\alpha_{\nr}}
    e^{-\mu \|\wv\|^2} \prod_{i=1}^{\nr} |w_i|^{2(\alpha_i - 1 )}, 
    \label{eq:gamma}
  \end{align}%
  where $g_{\pmb{\alpha}}$ is the normalization factor which only
  depends on $\pmb{\alpha}$ and ${\nr}$. Essentially, we let each $W_i$
  be independent and circularly symmetric with the squared magnitude
  following a single-variate Gamma distribution with parameter~$(\mu,\alpha_i)$, as 
  defined in~\eqref{eq:sv-gamma} from Definition~\ref{def:mv-Gamma}.

  \begin{lemma}
    \label{lemma:logq2}
    By choosing $0<\alpha_i<1$ and $\mu = \min\{P^{-1},1\}$, we have
    \begin{align}
      \MoveEqLeft[0]{\E\bigl[ - \log q({\rvVec{W}})
      \bigr] } \nonumber \\
      &\le \sum_{i=1}^{\nr} \alpha_i \log^+\! P + \sum_{i=1}^{\nr}
  2(1-\alpha_i) \E \bigl[ \log^+\! |\tilde{X}_i| \bigr]  + \constH,
  \label{eq:Elogq2}
    \end{align}%
    where we assume that $|\tilde{X}_1|\ge\ldots\ge|\tilde{X}_{\nt}|$
    for notational convenience.
 \end{lemma}
 \begin{proof}
   The following is straightforward from~\eqref{eq:gamma},
    \begin{align}
      \MoveEqLeft[0.5]{\E\bigl[ - \log q({\rvVec{W}})
      \bigr] } \nonumber \\
      &\le \sum_{i=1}^{\nr} \alpha_i \log \mu^{-1} + \mu \E\bigl[ \|\rvVec{W} \|^2 \bigr]
      \log e \nonumber \\
      &\qquad + \sum_{i=1}^{\nr} 2(1-\alpha_i) \E \bigl[ \log |{W}_i| \bigr]  + \const \\
      &\le \sum_{i=1}^{\nr} \alpha_i \log^+\! P  
      + \sum_{i=1}^{\nr} (1-\alpha_i) \E \bigl[ \log |{W}_i|^2 \bigr]  + \constH'.
      \label{eq:tmp7543}
    \end{align}%
    Note that $|W_i|^2 \le 2 |\tilde{X}_i|^2 +
    2\sigma_{\max}^2(\pmb{H}_U^{-1}\pmb{H}_{\bar{U}}) \| \tilde{\rvVec{X}}_{\bar{U}}\|^2
    + 2 |Z_i|^2$. From the definition of $U$, we have $\|
    \tilde{\rvVec{X}}_{\bar{U}}\|^2 \le (\nt-{\nr}) |\tilde{X}_i|^2$,
    $\forall\, i\in U$. Thus,
    $|W_i|^2  \le (2+|\constH''|) |\tilde{X}_i|^2  + 2 |Z_i|^2$. Applying Jensen's
    inequality on the expectation over $Z_i$, we get $\E \bigl[ \log |{W}_i|^2 \bigr]
    \le \E \bigl[ \log ( (2+|\constH''|) |\tilde{X}_i|^2  + 2 ) \bigr] 
    \le 2 \E \bigl[ \log^+\! |\tilde{X}_i| \bigr] + \constH'''$. Plugging it back to
    \eqref{eq:tmp7543}, we readily have \eqref{eq:Elogq2}. 
 \end{proof}

 Finally, putting together \eqref{eq:tmp8222} and \eqref{eq:Elogq2}, we obtain
\begin{align}
  \MoveEqLeft[0]{I( \tilde{\rvVec{X}}; {\rvVec{W}} \cond U)} \nonumber \\
  &\le \sum_{i=1}^{\nr} \alpha_i \log^+\! P +  \sum_{i=i}^{\nr} (1-2\alpha_i)  \E \bigl[ \log^+\!
|\tilde{X}_i| \bigr]  + \constH \label{eq:tmp9119} \\
&\le \sum_{i=1}^{\nr} \alpha_i \log^+\! P +  \sum_{i=i}^{\nr} (1-2\alpha_i)^+ \, \E \bigl[
\log^+\! |\tilde{X}_i| \bigr]  + \constH \label{eq:tmp9118} \\
&\le  \frac{1}{2} \sum_{i=1}^{\nr} \left(2\alpha_i+ (1-2\alpha_i)^+\right)
\log^+\! P  + \constH',
    \end{align}
    where, to obtain the last inequality, we apply \eqref{eq:Jensen} in Lemma~\ref{lemma:log+} with
    $p=2$, and the power constraint $\E \bigl[ |\tilde{X}_i|^2 \bigr] \le a^2 \E \bigl[
    \|\rvVec{X}\|^2 \bigr] \le a^2 P$. 
Therefore, the multiplexing gain is upper-bounded by 
\begin{align}
  \frac{1}{2} \sum_{i=1}^{\nr} \left(2\alpha_i+ (1-2\alpha_i)^+\right), \quad \forall\,
  \pmb{\alpha}\in(0,1)^{\nr}.
\end{align}%
Taking the infimum over $\pmb{\alpha}$, we get
$\frac{\nr}{2}$.

\subsection{Case~B3: Receive phase noise}

First  it is not hard to show the upper bound $\nt-\frac{1}{2}$. It is enough to provide the
$\nr-1$ relative angles, $\{\Theta_{\text{R},k} - \Theta_{\text{R},1}\}_{k=2\ldots \nr}$, to the
receiver. The channel is now equivalent to the case with common phase noise
$\Theta_{\text{R},1}$. Then  we can apply Proposition~\ref{prop:1}, since
\begin{align}
  \MoveEqLeft{h(\Theta_{\text{R},1} \cond \{\Theta_{\text{R},k} -
  \Theta_{\text{R},1}\}_{k=2\ldots \nr})}\nonumber \\
  &= 
  h(\rvVec{\Theta}_{\text{R}}) - h(\{\Theta_{\text{R},k} - \Theta_{\text{R},1}\}_{k=2\ldots
  \nr})
  \\
  &\ge h(\rvVec{\Theta}_{\text{R}}) - (\nr-1)\log(2\pi) \\
  &\ge -\infty.
\end{align}%
Next, we show the upper bound $\frac{\nr}{2}$. To that end, we write 
\begin{align}
  \MoveEqLeft[1]{I(\rvVec{X}; \rvVec{Y})} \nonumber \\
  &= I(\pmb{H}\rvVec{X}; |\rvVec{Y}|) + I(\pmb{H}\rvVec{X}; \angle_{\rvVec{Y}} \cond |\rvVec{Y}|) \\
  &= I(\pmb{H}\rvVec{X}; |\rvVec{Y}|) + I(\pmb{H}\rvVec{X};
  (\angle_{\pmb{H}\rvVec{X}+\rvVec{Z}'} +
  \rvVec{\Theta}_{\text{R}})_{2\pi} \cond |\rvVec{Y}|) \\
  &= I(\pmb{H}\rvVec{X}; |\rvVec{Y}|) + h(
  (\angle_{\pmb{H}\rvVec{X}+\rvVec{Z}'} + \rvVec{\Theta}_{\text{R}} )_{2\pi}
  \cond |\rvVec{Y}|) \nonumber \\
  & \qquad -  h( (\angle_{\pmb{H}\rvVec{X}+\rvVec{Z}'} + \rvVec{\Theta}_{\text{R}} )_{2\pi} \cond |\rvVec{Y}|, \pmb{H}\rvVec{X}) \\
  &\le I(\pmb{H}\rvVec{X}; |\rvVec{Y}|) + \nr\log(2\pi) \nonumber \\
  &\qquad -  h( (\angle_{\pmb{H}\rvVec{X}+\rvVec{Z}'} + \rvVec{\Theta}_{\text{R}})_{2\pi} \cond |\rvVec{Y}|, \pmb{H}\rvVec{X}, \angle_{\pmb{H}\rvVec{X}+\rvVec{Z}'}) \\
  &= I(\pmb{H}\rvVec{X}; |\rvVec{Y}|) + \nr\log(2\pi) -  h(  \rvVec{\Theta}_{\text{R}}),
\end{align}%
where we define $\rvVec{Z}'\defeq
e^{-j\rvVec{\Theta}_{\text{R}}} \circ \rvVec{Z}$ which is independent
of $\rvVec{\Theta}_{\text{R}}$ since $\rvVec{Z}$ is circularly
symmetric; the last equality follows since $\rvVec{Y} = e^{j\rvVec{\Theta}_{\text{R}}} \circ
(\pmb{H}\rvVec{X} + \rvVec{Z}')$ and thus $\rvVec{\Theta}_{\text{R}}$
is independent of $(|\rvVec{Y}|, \pmb{H}\rvVec{X} +
\rvVec{Z}', \pmb{H}\rvVec{X})$. 
It remains to show that $
I(\pmb{H}\rvVec{X}; |\rvVec{Y}|) \le \frac{\nr}{2} \log^+\! P + \constH$. To prove this, it
is enough to apply 
$h(|\rvVec{Y}|) \le \frac{\nr}{2} \log^+\! P + \constH$ and to use the fact that
$h(|\rvVec{Y}| \cond \pmb{H}\rvVec{X}) = \sum_{k=1}^{\nr} h(|Y_k| \cond \pmb{H}\rvVec{X})$
is lower-bounded by some constant according to \eqref{eq:beta} in Lemma~\ref{lemma:ejt}.

\section{Capacity Lower Bound for Model~B}
\label{sec:LB}

In this section, we derive a lower bound on the capacity of model~B. For simplicity, we
consider the class of memoryless Gaussian input distributions. 
Although the optimal input
distribution has been proved to be discrete in~\cite{Katz}, the use of a simple Gaussian
input provides tight lower bounds on the pre-log, which is enough for our purpose here. 
In the following, we only consider the
memoryless phase noise channel which can be shown to have a lower capacity than the
general stationary and ergodic channel when memoryless input is used. To see this, we
write   
\begin{align}
  I(\rvVec{X}^N; \rvVec{Y}^N) &= h(\rvVec{X}^N) - h(\rvVec{X}^N \cond \rvVec{Y}^N) \\
  &= \sum_{t=1}^{N} h(\rvVec{X}_t) - \sum_{t=1}^{N} h(\rvVec{X}_t \cond \rvVec{X}^{t-1}, \rvVec{Y}^N) \\
  &\ge \sum_{t=1}^{N} h(\rvVec{X}_t) - \sum_{t=1}^{N} h(\rvVec{X}_t \cond \rvVec{Y}_t) \\
  &= \sum_{t=1}^{N} I(\rvVec{X}_t; \rvVec{Y}_t) \\
  &= {N} I(\rvVec{X}; \rvVec{Y}). 
\end{align}%
Thus, we focus on the single-letter mutual information~$I(\rvVec{X};\rvVec{Y})$ in the rest of the
section.  As in the previous section, we investigate the three cases separately. 

\subsection{Case~B1: Transmit and receive phase noises}
In this case, we use all the inputs with equal power, i.e.,  
$\rvVec{X}\sim\mathcal{CN}(0,\frac{P}{\nt}\Id_{\nt})$. For convenience, let us rewrite the received signal as
\begin{align}
  \rvVec{Y} &= e^{j{\rvVec{\Theta}}_{\text{R}}} \circ \bigl( \pmb{H} (e^{j{\rvVec{\Theta}}_{\text{T}}} \circ \rvVec{X} ) \bigr) + \rvVec{Z} \\
  &= \sqrt{\frac{P}{\nt}} e^{j\tilde{\rvVec{\Theta}}_{\text{R}}} \circ \bigl( \pmb{H} (e^{j\tilde{\rvVec{\Theta}}_{\text{T}}} \circ \rvVec{X}_0 ) \bigr) + \rvVec{Z} \\
  &= \sqrt{\frac{P}{\nt}} e^{j\tilde{\rvVec{\Theta}}_{\text{R}}} \circ \hat{\rvVec{Y}} +
  \rvVec{Z} = \sqrt{\frac{P}{\nt}} \tilde{\rvVec{Y}} + \rvVec{Z},  
\end{align}%
where $\rvVec{X}_0\sim\mathcal{CN}(0,\Id_{\nt})$ is the normalized version of $\rvVec{X}$; $\tilde{\rvVec{\Theta}}_{\text{R}}\defeq {\rvVec{\Theta}}_{\text{R}} + \Theta_{\text{T},1}$
and $\tilde{\rvVec{\Theta}}_{\text{T}}\defeq {\rvVec{\Theta}}_{\text{T}} - \Theta_{\text{T},1}$. 
Note that $\tilde{\Theta}_{\text{T},1} = 0$ by definition and
$h(\tilde{\rvVec{\Theta}}_{\text{R}})>-\infty$. 
The mutual information of interest can be written as 
\begin{align}
  {I(\rvVec{X}; \rvVec{Y})} 
  &= I(\rvVec{X}, \tilde{\rvVec{\Theta}}_{\text{T}}; \rvVec{Y}) - I(\tilde{\rvVec{\Theta}}_{\text{T}}; \rvVec{Y} \cond \rvVec{X}) \\
  &= h(\rvVec{Y}) - h(\rvVec{Y} \cond \rvVec{X}, \tilde{\rvVec{\Theta}}_{\text{T}}) - I(\tilde{\rvVec{\Theta}}_{\text{T}}; \rvVec{Y} \cond
  \rvVec{X} ). \label{eq:tobound}
\end{align}%
First the following lemma, which provides a lower bound on $h(\rvVec{Y})$ in
\eqref{eq:tobound}, is crucial for the achievability proof. 
\begin{lemma}
\label{lemma:essential}
With receive phase noise such that $h(\rvVec{\Theta}_{\text{R}}) >
-\infty$, we have 
\begin{align}
  h(\rvVec{Y}) &\ge 
  \Bigl(\frac{\nr}{2} + \frac{1}{2} \min\{\nr, 2\nt-1\}\Bigr) \log^+\! P + \constH. 
  \label{eq:hY}
\end{align}%
\end{lemma}
\begin{proof}
  See Appendix~\ref{app:essential}. 
\end{proof}
Next, we derive upper bounds on the two negative terms in \eqref{eq:tobound} as follows. The conditional differential entropy can be upper-bounded as
\begin{align}
  {h(\rvVec{Y} \cond \rvVec{X}, \tilde{\rvVec{\Theta}}_{\text{T}}) }
  &\le \sum_{k=1}^{\nr} h(Y_k \cond \rvVec{X}, \tilde{\rvVec{\Theta}}_{\text{T}}) \label{eq:tmp120}\\
  &\le \sum_{k=1}^{\nr} \mathbb{E} \left[ \log^+\! \biggl| \sqrt{\frac{P}{\nt}}\pmb{h}_{k}^\T
   (e^{j\tilde{\rvVec{\Theta}}_{\text{T}}} \circ \rvVec{X}_0 ) \biggr| \right] + \const \\ 
   &\le \frac{\nr}{2} \log^+\! P + \constH, \label{eq:tmp914}
\end{align}
where  the second inequality is due to Lemma~\ref{lemma:ejt} and the third inequality is from
 Lemma~\ref{lemma:log+} and the power constraint $\E[\bigl| \pmb{h}_{k}^\T (e^{j\tilde{\rvVec{\Theta}}_{\text{T}}} \circ \rvVec{X}_0 ) \bigr|^2] =
 \nt^{-1} \|\pmb{h}_{k}\|^2 P$.
And 
\begin{align}
  {I(\tilde{\rvVec{\Theta}}_{\text{T}}; \rvVec{Y} \cond \rvVec{X})} 
  &\le I(\tilde{\rvVec{\Theta}}_{\text{T}}; \rvVec{Y}, \tilde{\rvVec{\Theta}}_{R} \cond \rvVec{X}) \label{eq:tmp121}\\
  &= I(\tilde{\rvVec{\Theta}}_{\text{T}}; \rvVec{Y} \cond
  \rvVec{X}, \tilde{\rvVec{\Theta}}_{R}) + I(\tilde{\rvVec{\Theta}}_{\text{T}}; \tilde{\rvVec{\Theta}}_{R}) \label{eq:tmp1212}\\
  &= I(\tilde{\rvVec{\Theta}}_{\text{T}}; e^{j{\tilde{\rvVec{\Theta}}}_{\text{T}}} \circ \rvVec{X} +
  \pmb{H}^\dag \rvVec{Z} \cond \rvVec{X}, \tilde{\rvVec{\Theta}}_{R}) + \const\\
  &\le I(\tilde{\rvVec{\Theta}}_{\text{T}}; e^{j{\tilde{\rvVec{\Theta}}}_{\text{T}}} \circ \rvVec{X} + \tilde{\rvVec{Z}} \cond \rvVec{X}, \tilde{\rvVec{\Theta}}_{R}) + \const\\
  &= h(e^{j{\tilde{\rvVec{\Theta}}}_{\text{T}}} \circ \rvVec{X} + \tilde{\rvVec{Z}} \cond
  \rvVec{X}, \tilde{\rvVec{\Theta}}_{R}) \nonumber \\
  &\quad - h(e^{j{\tilde{\rvVec{\Theta}}}_{\text{T}}} \circ \rvVec{X} + \tilde{\rvVec{Z}} \cond \rvVec{X}, \tilde{\rvVec{\Theta}}_{R}, \tilde{\rvVec{\Theta}}_{\text{T}}) + \const\\
   &\le \frac{\nt-1}{2} \log^+\! P + \constH, \label{eq:tmp915}
\end{align}%
where $\tilde{\rvVec{Z}} \sim \mathcal{CN}(0, \sigma^2_{\min}(\pmb{H}^\dag)
\Id_{\nt})$, with $\sigma_{\min}(\pmb{H}^\dag)$ being the minimum singular value of $\pmb{H}^\dag$; to obtain
the last inequality, we use the fact that $\tilde{\Theta}_{\text{T},1}=0$ and apply
Lemma~\ref{lemma:ejt} then Lemma~\ref{lemma:log+} for the rest of the $\nt-1$ entries of
$\tilde{\rvVec{\Theta}}_{\text{T}}$.

Plugging \eqref{eq:tmp914}, \eqref{eq:tmp915}, and \eqref{eq:hY} into \eqref{eq:tobound}, we obtain
\begin{align}
  I(\rvVec{X}; \rvVec{Y}) &\ge \frac{1}{2} \min\{ {\nr-\nt+1}, {\nt} \} \log^+\! P + \constH. 
\end{align}%
Note that the above lower bound holds when we substitute $\nt$ by any $\nt'\le \nt$, i.e., by activating
only $\nt'$ transmit antennas. It is clear that when $\nr-\nt+1\ge \nt$, i.e., $\nr\ge2\nt-1$, we should let
$\nt'=\nt$. Otherwise, we should decrease $\nt'$ to balance between $\nr-\nt'+1$ and $\nt'$, which gives
$\nt'=\lfloor \frac{\nr+1}{2} \rfloor$. This completes the proof of the lower bound for model~B1.

\subsection{Case~B2: Transmit phase noise}

In this case, we use $\nt'\defeq \min\{\nt,\nr\}$ input antennas and deactivate the remaining
ones. The active inputs, denoted by $\rvVec{X}'$, are i.i.d.~Gaussian, i.e.,
$\rvVec{X}'\sim\mathcal{CN}(0,\frac{P}{\nt'}\Id_{\nt'})$. We rewrite the output vector as
$\rvVec{Y} = \pmb{H}' (e^{j\rvVec{\Theta}'_{\text{T}}} \circ \rvVec{X}') +
\rvVec{Z}$ where $\pmb{H}'\in\mathbb{C}^{\nr\times{\nt'}}$ is the submatrix of
$\pmb{H}$ corresponding to the active inputs, and $\rvVec{\Theta}'_{\text{T}}$ is similarly defined. 
It follows that $I(\rvVec{X}'; \rvVec{Y}) =
I(\rvVec{X}'; (\pmb{H}')^{\dag} \rvVec{Y})$. Then  we have $h( (\pmb{H}')^{\dag} \rvVec{Y}) =
h(e^{j\rvVec{\Theta}'_{\text{T}}} \circ \rvVec{X}' + (\pmb{H}')^{\dag}\rvVec{Z}) = \nt' \log^+\! P
+ \constH$ and $h( (\pmb{H}')^{\dag} \rvVec{Y} \cond \rvVec{X}') = h(e^{j\rvVec{\Theta}'_{\text{T}}}
\circ \rvVec{X}' + (\pmb{H}')^{\dag}\rvVec{Z} \cond \rvVec{X}') \le h(e^{j\rvVec{\Theta}'_{\text{T}}}
\circ \rvVec{X}' + \sigma_{\max}\bigl((\pmb{H}')^{\dag}\bigr)\rvVec{Z} \cond \rvVec{X}')$. The latter
is further upper-bounded by 
$\sum_{k=1}^{\nt'} \E \bigl[ \log^+\! |X_k| \bigr] + \constH \le \frac{\nt'}{2} \log^+\! P + \constH'$
according to \eqref{eq:Theta} in Lemma~\ref{lemma:ejt} and \eqref{eq:Jensen} in Lemma~\ref{lemma:log+}. This shows the lower
bound $\frac{1}{2} \min\{\nt,\nr\}$ on the multiplexing gain. 

\subsection{Case~B3: Receive phase noise}

As in Case~B1, we let $\rvVec{X}\sim\mathcal{CN}(0,\frac{P}{\nt}\Id_{\nt})$.
First $h(\rvVec{Y})$ is lower-bounded in Lemma~\ref{lemma:essential}.
Next, it readily follows from \eqref{eq:tmp120} to \eqref{eq:tmp914} that 
\begin{align}
  {h(\rvVec{Y} \cond \rvVec{X}) }
   &\le \frac{\nr}{2} \log^+\! P + \constH, \label{eq:tmp9914}
\end{align}%
since we are in the same situation as in Case~B1 when $\tilde{\rvVec{\Theta}}_{\text{T}}$ is known.  
Finally, combining \eqref{eq:hY} and \eqref{eq:tmp9914}, we obtain a lower bound on the mutual information 
\begin{align}
I(\rvVec{X}; \rvVec{Y}) &= h(\rvVec{Y}) - h(\rvVec{Y} \cond \rvVec{X}) \\
   &\ge \frac{1}{2} \min\{\nr, 2\nt-1\} \log^+\! P + \constH 
\end{align}%
which provides the desired multiplexing gain.

\section{Conclusions and Discussions}
\label{sec:conclusion}

In this work, we investigated the discrete-time stationary and ergodic $\nr\times\nt$ MIMO phase noise
channel. We characterized the exact multiplexing gain when phase noises are on the
individual paths and when phase noises are at either side of the channel. With both transmit and
receive phase noises, upper and lower bounds have been derived. In
particular, the upper bound results in this paper have been obtained via
the duality using a newly introduced multi-variate Gamma distribution.

For model~B1, the upper and lower bounds derived in this paper do not match for
$\nr\in[4:2\nt-2]$. 
We conjecture that the upper
bound $\frac{1}{2}\min\left\{ \nt, \nr-1 \right\}$ is indeed loose. Let us recall that the upper bound
is obtained by lower-bounding $h(\rvVec{W}\cond\tilde{\rvVec{X}})$ with \eqref{eq:LB-B}, and by
upper-bounding $\E \bigl[-\log q(\rvVec{W})\bigr]$ with $q(\wv)$
being the multi-variate Gamma distribution. We
believe that both bounds are loose for model~B1 in general. First  we can show that 
\begin{multline}
   h({\rvVec{W}} \cond \tilde{\rvVec{X}}) \le \nr \,\E \bigl[\log^+\!| \tilde{X}_U|\bigr] + \E \bigl[ \log^+ \!
   |\tilde{X}_V| \bigr] \\ 
   + \sum_{k\not\in\{ U, V\}} \E \bigl[ \log^+ \! |\tilde{X}_k| \bigr] + \constH. \label{eq:UB-B}
\end{multline}%
To see this, we can first write $h({\rvVec{W}} \cond \tilde{\rvVec{X}}) = h({\rvVec{W}} \cond
\tilde{\rvVec{X}}, \tilde{\rvVec{\Theta}}_{\text{T}}) + I(\tilde{\rvVec{\Theta}}_{\text{T}}; {\rvVec{W}} \cond \tilde{\rvVec{X}})$, 
then upper-bound the first term with $\nr \,\E \bigl[ \log^+\! |\tilde{X}_U| \bigr] + \constH'$ by following
closely the steps as in \eqref{eq:tmp120}-\eqref{eq:tmp914}, and upper-bound the second term
with $\sum_{k\ne U} \E \bigl[ \log^+ \! |\tilde{X}_k| \bigr] + \constH''$ by following closely the steps as in
\eqref{eq:tmp121}-\eqref{eq:tmp915}. As compared to the lower bound
\eqref{eq:LB-B}, the upper bound \eqref{eq:UB-B} differs only in the terms involving $\tilde{X}_k$,
$k\not\in\left\{ U,V \right\}$. In the following, we argue that even if the lower bound
$h({\rvVec{W}} \cond \tilde{\rvVec{X}})$ was the RHS of
\eqref{eq:UB-B} -- which is the largest
that one could get as lower bound since it is also an upper bound -- we still would not be able to tighten the
multiplexing gain upper bound $\frac{1}{2}\min\left\{\nr-1 \right\}$ with the same choice of
auxiliary distribution $q(\wv)$. In other words, for the given $q(\wv)$, \eqref{eq:LB-B} is tight
enough with respect to the upper bound on  $\E \bigl[ -\log q(\rvVec{W})\} \bigr]  - h(\rvVec{W}\cond\tilde{\rvVec{X}})$. 
To prove this, it is enough to observe that $\E \bigl[ -\log q(\rvVec{W}) \bigr]$ does not involve any
terms of $\tilde{\rvVec{X}}$ other than $\tilde{X}_U$ and $\tilde{X}_V$ in such a way to change
the high SNR behavior, whereas $h(\rvVec{W}\cond\tilde{\rvVec{X}})$ is increasing with the
strength of each $\tilde{X}_k$. Therefore, the maximization of $\E \bigl[ -\log q(\rvVec{W})\bigr] -
h(\rvVec{W}\cond\tilde{\rvVec{X}})$ over $\tilde{\rvVec{X}}$ will always put all $\tilde{X}_k$,
$k\not\in\left\{ U,V \right\}$, to zero, even if $h(\rvVec{W}\cond\tilde{\rvVec{X}})$ hits the
highest value \eqref{eq:UB-B}.     
To sum up, if the current upper bound $\frac{1}{2}\min\left\{ \nt, \nr-1 \right\}$ was indeed loose as
we conjecture, one would have to first find a new auxiliary
distribution $q(\wv)$ in order to get a tighter upper bound. In particular, the new auxiliary
distribution should be such that $\E \bigl[ -\log q(\rvVec{W})\bigr]$ depends on $\tilde{X}_k$,
$k\not\in\left\{ U,V \right\}$ at high SNR in a non-trivial way. With such a distribution, the
second challenge is to find a lower abound on $h({\rvVec{W}} \cond \tilde{\rvVec{X}})$ that also
depends on $\tilde{X}_k$, $k\not\in\left\{ U,V \right\}$,  in a non-trivial way.
In fact, we conjecture that  \eqref{eq:UB-B} holds with equality.

For model~B2 and B3, the results have the following alternative chain rule interpretation.
With transmit phase noise~(model~B2), the mutual information can be written as
$I(\rvVec{X}; \rvVec{Y}) =  I(\rvVec{X},\rvVec{\Theta}_{\text{T}}; \rvVec{Y} ) - I(
\rvVec{\Theta}_{\text{T}}; \rvVec{Y} \cond \rvVec{X})$, where the first term scales as
$\min\left\{ \nt, \nr \right\} \log P$ as if the phase noise were part of the transmitted signal
whereas the second part scales as $\frac{1}{2}\min\left\{ \nt, \nr \right\} \log P$ as if
$\rvVec{\Theta}$ were the input with a fixed distribution and $\rvVec{X}$ were the
``fading'' known at the receiver side. With receive phase noise~(model~B), the mutual
information can be written differently as $I(\rvVec{X}; \rvVec{Y}) =  I(\rvVec{X};
\rvVec{Y} \cond \rvVec{\Theta}_{\text{R}}) - I( \rvVec{\Theta}_{\text{R}}; \rvVec{X} \cond
\rvVec{Y})$. Here the first term corresponds to the rate when the phase noise is known,
while the second term can be considered as the rate of a ``reverse'' channel with input
$\rvVec{\Theta}_\text{R}$, output $\rvVec{X}$, and known fading $\rvVec{Y}$. In both cases,
the original problem of characterizing $I(\rvVec{X}; \rvVec{Y})$ boils down to subproblems
involving channels without phase noise~(i.e., $I(\rvVec{X},\rvVec{\Theta}_{\text{T}}$ and
$I(\rvVec{X}; \rvVec{Y} \cond \rvVec{\Theta}_{\text{R}})$) and communications with fixed
phase signaling~(i.e., $I(\rvVec{\Theta}_{\text{T}}; \rvVec{Y} \cond \rvVec{X})$ and $I(
\rvVec{\Theta}_{\text{R}}; \rvVec{X} \cond \rvVec{Y})$).

There are a few  interesting future directions. First, it is possible to
extend the results to multi-user channels and study the impact of phase
noise to such systems. Second, the lower bound on model~B1 suggests the
following dimension counting argument: one can recover $\nt$ real
information with $2\nt-1$ real observations, since the remaining $\nt-1$
dimensions are occupied by the $\nt-1$ relative phase noises. How to
design decoding algorithms that ``solve'' efficiently the $2\nt-1$
\emph{non-linear} equations is a question of both theoretical and practical importance.
Finally, a more refined analysis should lead to tighter upper and lower
bounds on the capacity, beyond the pre-log characterization.


\appendix

\subsection{Proof of Proposition~\ref{prop:1}} \label{app:prop1}
  With common phase noise, we can perform unitary precoding without losing information, and
  the channel is equivalent to a parallel channel with common phase noise $\rvVec{Y}_t =
  e^{j\Theta_t} \pmb{\Sigma}\, \pmb{x}_t + \rvVec{Z}_t = e^{j\Theta_t} \pmb{\sigma} \circ
  \pmb{x}_t + \rvVec{Z}_t$,
  where $\pmb{\Sigma}$ is a diagonal matrix with the $\min\{\nt,\nr\}$ non-zero singular values of
  the matrix $\pmb{H}$ and $\pmb{\sigma}$ is a vector of these elements. From \cite{Durisi-capa}, we know that the multiplexing gain of
  a $M\times M$ channel is upper-bounded by $M - \frac{1}{2}$. This upper
  bound applies here with $M=\min\{\nt,\nr\}$. The lower bound is achieved by using the
  Gaussian memoryless input 
  $\rvVec{X}_t\sim\mathcal{CN}(0,\frac{P}{\nt}\Id_{\nt})$, from which we have
  $I(\rvVec{X}^N;\rvVec{Y}^N) = h(\rvVec{Y}^N) - h(\rvVec{Y}^N \cond \rvVec{X}^N)$ with
  $h(\rvVec{Y}^N) = N h(\rvVec{Y}) = N \min\{\nt,\nr\} \log^+\! P + N \constH$ and $h(\rvVec{Y}^N \cond
  \rvVec{X}^N) \le N h(\rvVec{Y} \cond \rvVec{X}) = N h(e^{j\Theta} \pmb{\sigma}\circ
  \rvVec{X} + \rvVec{Z} \cond \rvVec{X})$. Applying a unitary transformation on $e^{j\Theta}
  \pmb{\sigma}\circ \rvVec{X} + \rvVec{Z} $, we obtain
  $N h(e^{j\Theta} \pmb{\sigma}\circ \rvVec{X} + \rvVec{Z} \cond \rvVec{X})= N h(e^{j\Theta}
  \|\pmb{\sigma}\circ \rvVec{X}\| + {Z}'_1 \cond \rvVec{X}) + N \sum_{k=2}^M h(Z'_k) \le N \E \bigl[
  \log^+\! \|\pmb{\sigma}\circ \rvVec{X}\| \bigr] + N \constH' \le \frac{N}{2} \log^+\! P + N \constH'$ where $\rvVec{Z}'$ is the rotated version of $\rvVec{Z}$ and remains spatially white, the first inequality is from
  Lemma~\ref{lemma:ejt} and the second one is from Lemma~\ref{lemma:log+}. Finally, we have
  $\frac{1}{N} I(\rvVec{X}^N;\rvVec{Y}^N) \ge \bigl(\min\{\nt,\nr\}-\frac{1}{2}\bigr) \log^+\! P + \constH''$, which completes the proof. 

\subsection{Proof of Lemma~\ref{lemma:LB-A} and \ref{lemma:LB-B}}
\label{app:LB}

In the following we shall derive the lower bounds~\eqref{eq:LB-A} and \eqref{eq:LB-B} on the
conditional differential entropy $h(\rvVec{W} \cond \tilde{\rvVec{X}} )$ for model~A and
model~B1, respectively. 

First we shall show that, for both models, 
\begin{align}
  h(W_i \cond \tilde{\rvVec{X}}) &\ge \E \bigl[ \log^+\!
  |\tilde{X}_U| \bigr]
  + \E \bigl[ \log^+\! |\tilde{X}_V| \bigr] + \constH. \label{eq:tmp544}
\end{align}%
To that end, we analyze
$h(W_i \cond \tilde{\rvVec{X}}=\pmb{x})$ with $|x_1|>
|x_2|> \cdots > |x_{\nt}| \ge 0$ without loss of generality,
i.e., we assume that $U=1$ and $V=2$. A lower bound of $h(W_i \cond
\tilde{\rvVec{X}}=\pmb{x})$ can be obtained by considering
the following cases separately. 
\begin{itemize}
  \item When $|x_1|\le 1$,
\begin{align}
  h(W_i \cond \tilde{\rvVec{X}} = \pmb{x}) &\ge h(W_i
  \cond \tilde{\rvVec{X}} = \pmb{x}, \rvMat{\Theta}) \\ 
  &= h(Z_i) \\
  &= \log(\pi e).
\end{align}%
  \item When $|x_1|\ge 1$ and $|x_2|\le 1$,
\begin{align}
  \MoveEqLeft
 {h(W_i \cond \tilde{\rvVec{X}} = \pmb{x})} \nonumber \\
 &\ge h(W_i \cond \tilde{\rvVec{X}} = \pmb{x}, \Theta_{i,2}, \ldots, \Theta_{i,\nt}) \\ 
  &= h(g_{i1} e^{j{\Theta_{i,1}}} x_1 + Z_i \cond \Theta_{i,2}, \ldots, \Theta_{i,\nt}) \\
  &\ge \log^+\! |g_{i1}x_1| + \const \\
  &\ge \log^+\! |x_1| + \constH,
\end{align}%
where 
$g_{i1} e^{j\Theta_{i,1}}$ is from the matrix
$\rvMat{G}_{(1)}$ defined in \eqref{eq:canonical}
since $U=1$ by assumption; the second inequality is from
Lemma~\ref{lemma:ejt} and the third inequality is from Lemma~\ref{lemma:log+}.
\item When $|x_1|\ge 1$ and $|x_2|\ge 1$,
    \begin{align}
      \MoveEqLeft[1]{h(W_i \cond \tilde{\rvVec{X}} = \pmb{x})} \nonumber \\
      &\ge {h(W_i \cond \tilde{\rvVec{X}} = \pmb{x}, \Theta_{i,3}, \ldots, \Theta_{i,\nt}, Z_i)}  \\
      &= h(e^{j\Theta_{i,1}} g_{i1} x_1 + e^{j\Theta_{i,2}} g_{i2} x_2 \cond \Theta_{i,3},
      \ldots, \Theta_{i,\nt}) \\
      &\ge h(e^{j\Theta_{i,1}} g_{i1} x_1 + e^{j\Theta_{i,2}} g_{i2} x_2 \cond \Theta_{i,3}, \ldots,
      \Theta_{i,\nt}, \Omega) \\
      &= \E \log(|g_{i1}g_{i2}x_1 x_2 \sin(\Theta_{i,1}-\Theta_{i,2} + \phi)|) \nonumber \\
      &\qquad + h(\Theta_{i,1},\Theta_{i,2} \cond \Theta_{i,3}, \ldots, \Theta_{i,\nt}, \Omega)
      \label{eq:tmp888} \\
      &\ge \log |x_1|  + \log |x_2|  \nonumber \\
      &\qquad + \E \log(|\sin(\Theta_{i,1}-\Theta_{i,2} + \phi)|) +
      \constH \\
      &\ge \log^+\! |x_1|  + \log^+\! |x_2|  + \constH',
    \end{align}%
    where the first inequality is from conditioning reduces entropy; we partition $[0,2\pi)^2$
    in such a way that $e^{j\Theta_{i,1}} g_{i1} x_1 + e^{j\Theta_{i,2}} g_{i2} x_2$ is a
    bijective function of $(\Theta_{i,1},\Theta_{i,2})$ in each partition indexed by $\Omega$
    which takes a finite number of values; then we applied the change of variables from
    Lemma~\ref{lemma:h-prod} and obtain \eqref{eq:tmp888} with $\phi \defeq \angle_{g_{i1}x_1} -
    \angle_{g_{i2}x_2}$; finally, we use the fact that $|g_{i1} g_{i2}|$ is bounded for almost every
    $\pmb{H}$ and the application of Lemma~\ref{lemma:Elogsin} to get the last inequality. Note that
    $\log|x_k| = \log^+\! |x_k|$ for $k=1,2$ by assumption. 
\end{itemize}
Combining the three cases above and taking expectation
over $\tilde{\rvVec{X}}$, we get \eqref{eq:tmp544}.

\subsubsection{Proof of the lower bound \protect\eqref{eq:LB-A} for model~A}
For model~A, we have 
\begin{align}
  {h(\rvVec{W} \cond \tilde{\rvVec{X}})} 
  &= \sum_{i=1}^{\nr} h(W_i \cond \tilde{\rvVec{X}}, W^{i-1}) \\
  &\ge \sum_{i=1}^{\nr} h(W_i \cond \tilde{\rvVec{X}}, W^{i-1},
  \{\Theta_{l,1},\ldots,\Theta_{l,\nt}\}_{l<i}) \\
  &= \sum_{i=1}^{\nr} h(W_i \cond \tilde{\rvVec{X}}, \{\Theta_{l,1},\ldots,\Theta_{l,\nt}\}_{l<i})
  \label{eq:tmp920}\\
  &\ge {\nr} \E \bigl[ \log^+\! |\tilde{X}_U| \bigr] + {\nr} \E \bigl[ \log^+\! |\tilde{X}_V| \bigr] + \constH,
\end{align}%
where \eqref{eq:tmp920} is from the fact that $W_i$ only depends on $W^{i-1}$ through the input
$\tilde{\rvVec{X}}$ and the phase noises $\{\Theta_{l,1},\ldots,\Theta_{l,\nt}\}_{l<i}$;  
where the last inequality is from a modified version of \eqref{eq:tmp544} by introducing
$\{\Theta_{l,1},\ldots,\Theta_{l,\nt}\}_{l<i}$ in the condition. 

\subsubsection{Proof of the lower bound \protect\eqref{eq:LB-B} for model~B1}
For model~B1, we write
\begin{align}
  h(\rvVec{W} \cond \tilde{\rvVec{X}}) &= 
  h(W_1 \cond \tilde{\rvVec{X}}) + \sum_{i=2}^{\nr} h(W_i
  \cond \tilde{\rvVec{X}}, W^{i-1}),
  \label{eq:tmp892}
\end{align}%
where, according to \eqref{eq:tmp544}, the first term
is lower-bounded by 
\begin{align}
  h(W_1 \cond \tilde{\rvVec{X}}) \ge \E\bigl[  \log^+\! |\tilde{X}_U| \bigr] +
  \E \bigl[ \log^+\! |\tilde{X}_V| \bigr] + \constH. \label{eq:tmp6789}
\end{align}%
In the following, we derive a lower bound on the
second term. Let $B_i \defeq \sum_{k=1}^{\nt} g_{ik} \tilde{X}_k
e^{j\Theta_{\text{T},k}}$ where $g_{ik}$ is the channel coefficient without phase noise from
the canonical form $U$ defined in \eqref{eq:canonical}. Then
\begin{align}
  \MoveEqLeft[1]{h(W_i \cond 
  \tilde{\rvVec{X}}, W^{i-1})} \nonumber \\
  &= h(e^{j\Theta_{\text{R},i}} B_i + Z_i \cond
  \tilde{\rvVec{X}}, W^{i-1}) \\
  &\ge h(e^{j\Theta_{\text{R},i}} B_i + Z_i \cond \tilde{\rvVec{X}}, W^{i-1}, B_i,
  \rvVec{\Theta}_{\text{T}}, \Theta_{\text{R}}^{i-1}) \\
  &= h(e^{j\Theta_{\text{R},i}} B_i + Z_i \cond B_i, \rvVec{\Theta}_{\text{T}}, \Theta_{\text{R}}^{i-1}) \\
  &\ge \E \bigl[ \log^+\! |B_i| \bigr] + \const \label{eq:tmp674}\\ 
  &= \E_{\tilde{\rvVec{X}}} \left[ \E\bigl[  \log^+\! |B_i| \cond \tilde{\rvVec{X}} \bigr] \right] +
  \const, \label{eq:tmp6774}
\end{align}
where the first inequality is from conditioning
reduces entropy; \eqref{eq:tmp674} is from Lemma~\ref{lemma:ejt}. The conditional expectation
can be lower-bounded as follows
\begin{align}
  \MoveEqLeft[0]{\E\bigl[ \log |B_i| \cond \tilde{\rvVec{X}} \bigr]}
  \nonumber \\
  &= \frac{1}{2} \E \left[ \log \left(\left|
  \sum_{k=1}^{\nt} |g_{ik} \tilde{X}_k|
  e^{j(\Theta_{\text{T},k}+\tilde{\Phi}_{ik})} \right|^2\right) \right] \\ 
  &\ge  \frac{1}{2}\inf_{\pmb{x}\in \mathbb{R}^{\nt}:\, \|\pmb{x}\|=1}  \E
  \left[  \log \left(\left|
  \sum_{k=1}^{\nt} |g_{ik}|  
  e^{j(\Theta_{\text{T},k}+\tilde{\Phi}_{ik})} x_k \right|^2\right) \right]
  \nonumber \\ 
  &\qquad + \log \|\tilde{\rvVec{X}}\|  \\
  &\ge  \frac{1}{2}\inf_{\pmb{x}\in \mathbb{R}^{\nt}:\, \|\pmb{x}\|=1}  \E
  \left[ \log \left(\left| \sum_{k=1}^{\nt} |g_{ik}|  
  \cos(\Theta_{\text{T},k}+\tilde{\Phi}_{ik}) x_k \right|^2\right) \right]
  \nonumber  \\ 
  &\qquad + \log \|\tilde{\rvVec{X}}\|  \\
  &\ge  \log \|\tilde{\rvVec{X}}\| + \constH \label{eq:tmp6899}  \\
  &\ge  \log |\tilde{X}_U| + \constH, \label{eq:tmp6799} 
\end{align}%
where $\tilde{\Phi}_{ik} \defeq \angle_{g_{ik}\tilde{X}_k}$; \eqref{eq:tmp6899} is
obtained by applying Lemma~\ref{lemma:Elognorm} with $\rvVec{V} \defeq \Bigl[|g_{ik}|  
\cos(\Theta_{\text{T},k}+\tilde{\Phi}_{ik})\Bigr]_k$ with $\E\bigl[ \|\rvVec{V}\|^2 \bigr] \le \sum_k |g_{ik}|^2 < \infty$ and 
  \begin{align}
    h(\rvVec{V}) &\ge h(\rvVec{V} \cond \tilde{\rvVec{\Phi}}) \\
    &= h(\cos(\rvVec{\Theta}_{\text{T}}+\tilde{\rvVec{\Phi}}) \cond \tilde{\rvVec{\Phi}}) + \sum_k \log|g_{ik}| \\
    &> -\infty,
  \end{align}%
  where the equality is the application of the change of variables from Lemma~\ref{lemma:h-prod}; the last inequality is from Lemma~\ref{lemma:Elogsin}.  
  From \eqref{eq:tmp6774} and \eqref{eq:tmp6799}, we get
  \begin{align}
    h(W_i \cond \tilde{\rvVec{X}}, W^{i-1})
    &= \E_{\tilde{\rvVec{X}}} \left[ \E\bigl[ \log^+\! |B_i| \cond \tilde{\rvVec{X}} \bigr] \right] + \const \\ 
    &\ge \E_{\tilde{\rvVec{X}}}  \left[ \left( \E\bigl[ \log |B_i| \cond \tilde{\rvVec{X}} \bigr]
    \right)^+ \right] + \const \\
    &\ge \E \left[ \bigl( \log |\tilde{X}_U|  + \constH \bigr)^+ \right] + \const \\
    &\ge \E \bigl[ \log^+\! |\tilde{X}_U| \bigr] - |\constH|  + \const, \label{eq:tmp6999}
  \end{align}%
  where the last inequality is from the application of \eqref{eq:logAX} in Lemma~\ref{lemma:log+} with $p=2^{\constH}$. 
Plugging \eqref{eq:tmp6789} and  
\eqref{eq:tmp6999} into \eqref{eq:tmp892}, the lower
bound \eqref{eq:LB-B} is obtained.  

\subsection{Proof of Lemma~\ref{lemma:logq}}
\label{app:logq}

From Definition~\ref{def:mv-Gamma}, by imposing
$1>\alpha_i>0$, $i=1,\ldots,\nr$, and $\mu= \min\{P^{-1},1\}$, we
have
\begin{align}
  - \log q(\rvVec{W}) &= -\log
  \frac{g_{\pmb{\alpha}}}{\nr!} - \sum_{i=1}^{\nr} \alpha_i \log \mu +
  \mu |\hat{W}_{\nr}|^2 \nonumber \\
  &\qquad + \sum_{i=2}^{\nr} (1-\alpha_i) \log\bigl( |\hat{W}_i|^2 - |\hat{W}_{i-1}|^2 \bigr)\nonumber \\
  &\qquad + (1-\alpha_1) \log|\hat{W}_1|^2.\label{eq:log-q} 
\end{align}%

We bound each term as follows.
\begin{itemize}
  \item The squared magnitude of each output 
\begin{align}
  |{W}_i|^2 
  &\le 2|\rvVec{G}_i^\T  \tilde{\rvVec{X}}|^2 + 2|Z_i|^2 \\
  &\le 2 \|\rvVec{G}_i\|^2 \|\tilde{\rvVec{X}}\|^2  + 2 |Z_i|^2 \label{eq:CS}
  \\
  &\le 2 \Hmax \|\tilde{\rvVec{X}}\|^2  + 2 \|\rvVec{Z}\|^2,
 \end{align}
 where $\rvVec{G}_i^\T$ is the $i$\,th row of the
 canonical matrix $\rvMat{G}_{(U)}$ defined in
 \eqref{eq:canonical}; \eqref{eq:CS}
 is due to Cauchy-Schwarz; and $\Hmax$ is defined as
 \begin{align}
   \Hmax &\defeq \max_{u=1,\ldots,\nt} \| \rvMat{G}_{(u)} \|^2. 
 \end{align}%
 
\item The difference of the squared magnitudes
\begin{align}
  \MoveEqLeft[1]{\left| |{W}_i|^2 - |{W}_k|^2 \right|} \nonumber\\
  &\le  (| {W}_i| + |{W}_k |) \,     |
  e^{-j\Theta_{i,U}} {W}_i - e^{-j\Theta_{k,U}}
  {W}_k  | \nonumber\\
  &\le 2^{\frac{3}{2}} \sqrt{ \Hmax
  \|\tilde{\rvVec{X}}\|^2  + \|\rvVec{Z}\|^2 } \,  |
 e^{-j\Theta_{i,U}} {W}_i - e^{-j\Theta_{k,U}}
 {W}_k | \nonumber  
\end{align}
with
\begin{align}
  \MoveEqLeft[0]{ 
 | e^{-j\Theta_{i,U}} {W}_i - e^{-j\Theta_{k,U}}
 {W}_k |^2 } \nonumber\\
  &\le {\left| e^{-j\Theta_{i,U}} ( \rvVec{G}_i^\T
  \tilde{\rvVec{X}} + Z_i) -  e^{-j\Theta_{k,U}} (
  \rvVec{G}_k^\T \tilde{\rvVec{X}} + Z_k) \right|^2} \nonumber \\
  &\le {2 \left| e^{-j\Theta_{i,U}} \rvVec{G}_i^\T
  \tilde{\rvVec{X}} -
  e^{-j\Theta_{k,U}} \rvVec{G}_k^\T \tilde{\rvVec{X}} \right|^2 + 2|
  Z_i|^2 + 2| Z_k |^2}\nonumber\\
  &\le 2 \biggl(\sum_{l\ne U} |e^{-j\Theta_{i,U}}
  G_{il} - e^{-j\Theta_{k,U}} G_{kl}|^2 \biggr)
  \biggl(\sum_{l\ne U} |\tilde{X}_l|^2 \biggr) \nonumber
  \\ &\qquad + 2\|\rvVec{Z}\|^2 \\
  &\le 4 (\nt-1) \Hmax \tilde{X}_V^2 + 2\|\rvVec{Z}\|^2.
\end{align}%
\end{itemize}
Note that the above upper bounds does not depend on
$i$ and $k$. Then, with the above bounds, we take
expectation of the terms in \eqref{eq:log-q}, and
obtain   
\begin{align}
  { \E \bigl[  |\hat{W}_{\nr}|^2
  \bigr] } 
  &\le 2  \bigl(  \Hmax \E \bigl[ \|\tilde{\rvVec{X}}\|^2 \bigr] + \E
  \bigl[ \|\rvVec{Z}\|^2 \bigr] \bigr) \\
  &\le 2 \Hmax a^2 P + \const, \label{eq:tmp821}\\
  { \E \bigl[ \log |\hat{W}_1|^2
  \bigr] } 
  &\le \E \bigl[ \log ( \Hmax \|\tilde{\rvVec{X}}\|^2 + \E \bigl[
  \|\rvVec{Z}\|^2 \bigr] ) \bigr] + 1\\
  &\le 2\, \E \bigl[ \log^+\! |\tilde{X}_U| \bigr] + \constH, \label{eq:tmp822} 
\end{align}
where the last inequality is from Lemma~\ref{lemma:log+}. Similarly, basic calculations lead to  
\begin{align}
  \MoveEqLeft[1.5]{\E \left[ \log\left| |{\hat{W}}_i|^2
  - |{\hat{W}}_{i-1}|^2 \right|\right] } \nonumber \\
  &\le \frac{1}{2} \E\left[ \log\left( 
  \Hmax \|\tilde{\rvVec{X}}\|^2   + \E \bigl[\|\rvVec{Z}\|^2 \bigr]
 \right) \right] \nonumber \\
 &\quad + \frac{1}{2} \E \left[ \log \left( 4 (\nt-1)
 \Hmax |\tilde{X}_V|^2 + 2 \E \bigl[ \|\rvVec{Z}\|^2
 \bigr] \right) \right]  + \const
\\
  &\le \log^+\! |\tilde{X}_U| + \log^+\! |\tilde{X}_V| + \constH.
   \label{eq:tmp823}
\end{align}%
Taking expectation over $\rvVec{X}$ in \eqref{eq:log-q}, and plugging
\eqref{eq:tmp821}, \eqref{eq:tmp822}, and
\eqref{eq:tmp823} into it, we readily obtain
\eqref{eq:Elogq}.

\subsection{Proof of Lemma~\ref{lemma:essential}}
\label{app:essential}

To prove Lemma~\ref{lemma:essential}, we deal with the cases $\nr=2\nt-1$ and $\nr\ne2\nt-1$ separately.  
Let us define $\hat{\rvVec{Y}}$ and $\tilde{\rvVec{Y}}$ such that 
\begin{align}
  \rvVec{Y}  
  &= \sqrt{\frac{P}{\nt}} e^{j\tilde{\rvVec{\Theta}}_{\text{R}}} \circ \hat{\rvVec{Y}} + \rvVec{Z} =
  \sqrt{\frac{P}{\nt}} \tilde{\rvVec{Y}} + \rvVec{Z}. 
\end{align}%

For notational convenience, we define $n\defeq \nr$ and $m\defeq \nt$ in the following
proof. 

\subsubsection{Case $n=2m-1$}

First  we show that \eqref{eq:hY} holds for $n=2m-1$.  We write $h(\rvVec{Y}) \ge h(\rvVec{Y} \cond \rvVec{Z}) = h(\rvVec{Y} - \rvVec{Z} \cond \rvVec{Z}) = h\Bigl(\sqrt{\frac{P}{m}}
\tilde{\rvVec{Y}}\Bigr) = n \log P + h(\tilde{\rvVec{Y}}) +  \const$. Now, it is enough to show
that $h(\tilde{\rvVec{Y}})>-\infty$. 
From Lemma~\ref{lemma:h-squared},
\begin{align}
  h(\tilde{\rvVec{Y}}) &\ge h(|\tilde{\rvVec{Y}}|^2) + \const \\
  &= h(|\hat{\rvVec{Y}}|^2) + \const \\
  &\ge h( \rvVec{S} \cond \hat{Y}_n) + h(|\hat{Y}_n|^2) + \const \\
  &= h( \rvVec{S} \cond \hat{Y}_n) + \constH, 
\end{align}%
where $\rvVec{S} \in \mathbb{R}^{n-1}$ with $S_i\defeq |\hat{Y}_i|^2$ for $i=1,\ldots,n-1$; 
the second inequality is from the
chain rule and that adding the condition on the phase of
$\hat{Y}_n$ reduces entropy; the last equality is due to $\hat{Y}_n \sim
\mathcal{CN}(0,m^{-1}\|\pmb{h}_n\|^2)$. Next we need to show that $h(\rvVec{S} \cond \hat{Y}_n)>-\infty$. 
Intuitively, given $\hat{Y}_n$,
$\rvVec{S}$ can be expressed as $n-1 = 2(m-1)$ real
functions of the $2(m-1)$ real random variables $\bigl(\Re\{\hat{Y}^{m-1}\}$,
$\Im\{\hat{Y}^{m-1}\}\bigr)$. Since
$h\bigl(\Re\{\hat{Y}^{m-1}\}, \Im\{\hat{Y}^{m-1}\}\bigr) = h(\hat{Y}^{m-1})$ is finite for almost
every $\pmb{H}$, as long as the mapping is not \emph{degenerated}, $h(\rvVec{S} \cond \hat{Y}_n)$
should be finite too. This argument is proved formally in the following. 

Since for any generic $\pmb{H}\in\mathbb{C}^{n\times m}$, any $m$ rows of the matrix
are linear independent, the remaining $n-m$ rows can be written as linear combinations of these
rows. Let us take the rows $\{1,2,\ldots,m-1,n\}$ of $\pmb{H}$. It readily follows that 
one can write 
\begin{align}
\hat{Y}_{m}^{n-1} = \pmb{B} \hat{Y}^{m-1} + \pmb{b} \hat{Y}_n
\end{align}%
with $\pmb{B}\in\mathbb{C}^{(m-1)\times(m-1)}$ and $\pmb{b}\in\mathbb{C}^{(m-1)\times1}$
depending only on $\pmb{H}$. Next let us partition the space $\mathbb{R}^{n-1}$
into a finite number of sets
in each one of which the mapping $(\Re\{\hat{{y}}^{m-1}\}, \Im\{{\hat{y}}^{m-1}\}) \mapsto
\pmb{s}$ is bijective. Note that this is possible since the mapping is quadratic in a
finite-dimensional space. Let $\Omega$ be the index of the partitions which only depends on
$\hat{Y}^{m-1}$. Then  
\begin{align}
  h(\rvVec{S} \cond \hat{Y}_n) &\ge h(\rvVec{S} \cond
  \hat{Y}_n, \Omega) \\
  &= \E\left[\log |\det(\rvMat{J})| \right] + h(\hat{Y}^{m-1} \cond \hat{Y}_n, \Omega) \\
  &= \E\left[\log |\det(\rvMat{J})| \right] + \constH,
\end{align}%
where the first equality is from Lemma~\ref{lemma:h-prod}; the second equality is due to 
$ h(\hat{Y}^{m-1} \cond \Omega, \hat{Y}_n) = h(\hat{Y}^{m-1} \cond \hat{Y}_n) - I(\Omega;
\hat{Y}^{m-1} \cond \hat{Y}_n)$ with $h(\hat{Y}^{m-1} \cond \hat{Y}_n) > -\infty$ for any generic
$\pmb{H}$ and  $I(\Omega; \hat{Y}^{m-1} \cond \hat{Y}_n) \le H(\Omega) < \infty$;
$\rvMat{J}$ is the Jacobian matrix with 
\begin{align}
  \det(\rvMat{J}) &= \det \begin{bmatrix}
    \frac{\partial\,{\rvVec{S}}}{\partial{\Re\{\hat{Y}^{m-1}\}}} & \frac{\partial\,{\rvVec{S}}}{\partial{\Im\{\hat{Y}^{m-1}\}}}\end{bmatrix}  \\
    &= 2^{m-1} \det \begin{bmatrix} \frac{\partial\,{\rvVec{S}}}{\partial{\hat{Y}^{m-1}}} &
      \frac{\partial\,{\rvVec{S}}}{\partial{(\hat{Y}^{m-1})^*}} \end{bmatrix} \label{eq:tmp4300}\\
    &= 2^{m-1} \det \begin{bmatrix} \diag\{(\hat{Y}^{m-1})^*\} & \diag\{\hat{Y}^{m-1}\} \\
      \diag\{(\hat{Y}_{m}^{n-1})^*\} \pmb{B}  &
      \diag\{\hat{Y}_{m}^{n-1}\} \pmb{B}^*  \end{bmatrix}  \label{eq:tmp4301}\\
      &= 4^{m-1} j^{m-1} \Im\left\{
    \diag\{\hat{Y}_{m}^{n-1}\}
    \pmb{B}^*\diag\{(\hat{Y}^{m-1})^*\} \right\},
    \label{eq:tmp211}
\end{align}%
where \eqref{eq:tmp4300} is due to the fact that the complex gradient of a real-valued function is
a unitary transformation of the real gradient~(see, e.g., \cite[App.A6]{Kailath}); to obtain the
last equality, we apply the identity $\det\left[\begin{smallmatrix} \pmb{C} & \pmb{D} \\ \pmb{E} & \pmb{F} \end{smallmatrix} \right] = \det(\pmb{C}) \det(\pmb{F} -
  \pmb{E}\pmb{C}^{-1}\pmb{D})$. 
  Since $\hat{Y}_1,\ldots,\hat{Y}_{m-1},\hat{Y}_n$ are jointly circularly symmetric Gaussian with
  finite and non-degenerate covariance for any generic $\pmb{H}$, there exists a $\hat{Y}'_n$
  circularly symmetric with non-zero bounded variance and independent of $\hat{Y}^{m-1}$, such that
  \begin{align}
    \hat{Y}_n &= \hat{Y}_n' + \pmb{f}^\T \hat{Y}^{m-1}, \quad\text{and thus}\\
    \hat{Y}_{m}^{n-1} &= (\pmb{B} + \pmb{b} \pmb{f}^\T) \hat{Y}^{m-1} + \pmb{b} \hat{Y}'_n,
  \end{align}%
  where $\pmb{f}\in\mathbb{C}^{(m-1)\times1}$ depends only on $\pmb{H}$. Then we can continue
  from \eqref{eq:tmp211} and write $|\det(\rvMat{J})|$ as
\begin{align}
  |\det(\rvMat{J})| &= 4^{m-1} \left|\det\Bigl(\Im\bigl\{ \hat{Y}_n' \diag\{\pmb{b}\} \rvMat{A}
  \bigr\} + \rvMat{M}  \Bigr)\right| \\
  &= 4^{m-1} \left| \det\bigl( \hat{Y}'_{n,\text{R}} \rvMat{N}_{\text{I}} +
  \hat{Y}'_{n,\text{I}} \rvMat{N}_{\text{R}} + \rvMat{M}  \bigr) \right|,
\end{align}%
where $\rvMat{A} \defeq \pmb{B}^*\diag\{(\hat{Y}^{m-1})^*\}$ and $\rvMat{M} \defeq
\Im \Bigl\{ \diag\{(\pmb{B}+ \pmb{b} \pmb{f}^\T)\hat{Y}^{m-1}\}
\pmb{B}^*\diag\{(\hat{Y}^{m-1})^*\} \Bigr\}$;
$\hat{Y}'_{n,\text{R}}$ and $\hat{Y}'_{n,\text{I}}$ are the real and imaginary parts of
$\hat{Y}_n'$, respectively; $\rvMat{N}_{\text{R}}$ and
$\rvMat{N}_{\text{I}}$ are the real and imaginary parts of
$\diag\{\pmb{b}\} \rvMat{A}$, respectively. Conditioned on
$\hat{Y}^{m-1}$ and $\hat{Y}'_{n,\text{I}}$, the determinant
$|\det(\rvMat{J})|$ can be written as the absolute value of
the characteristic polynomial of 
$(\hat{Y}'_{m,\text{I}} \rvMat{N}_{\text{R}} + \rvMat{M})
\rvMat{N}_{\text{I}}^{-1}$, namely,
\begin{align}
  |\det(\rvMat{J})| &= 4^{m-1}  |\det(\rvMat{N}_{\text{I}})| \prod_{t=1}^{m-1}
  |\hat{Y}'_{n,\text{R}} - \Lambda_t|,
\end{align}%
where $\Lambda_1,\ldots,\Lambda_{m-1}$ are the eigenvalues of $(\hat{Y}'_{n,\text{I}}
\rvMat{N}_{\text{R}} + \rvMat{M}) \rvMat{N}_{\text{I}}^{-1}$ in $\mathbb{C}$ and are functions
of $\pmb{H}$, $\hat{Y}^{m-1}$, and $\hat{Y}'_{n,\text{I}}$. Then
\begin{align}
  \MoveEqLeft[0]{\mathbb{E}\left[ \log |\det(\rvMat{J})| \cond
  \hat{Y}^{m-1}, \hat{Y}'_{n,\text{I}}\right]} \nonumber \\
  &= \mathbb{E} \bigl[ \log |\det(\rvMat{N}_{\text{I}})| \cond
  \hat{Y}^{m-1}, \hat{Y}'_{n,\text{I}}\bigr] \nonumber \\ &\quad +
  \sum_{t=1}^{m-1}  \mathbb{E} \bigl[ \log |\hat{Y}'_{n,\text{R}} -
  \Lambda_t| \cond \hat{Y}^{m-1}, \hat{Y}'_{n,\text{I}}\bigr] + \const \\
  &\ge \mathbb{E} \bigl[ \log |\det(\rvMat{N}_{\text{I}})| \cond
  \hat{Y}^{m-1}, \hat{Y}'_{n,\text{I}}\bigr] +
  \sum_{t=1}^{m-1}  \mathbb{E} \bigl[
  \log |\hat{Y}'_{n,\text{R}}| \bigr] + \const \\
  &= \mathbb{E} \bigl[ \log |\det(\rvMat{N}_{\text{I}})| \cond
  \hat{Y}^{m-1}, \hat{Y}'_{n,\text{I}}\bigr] + \constH,
\end{align}%
where the inequality is from $\log
|\hat{Y}'_{n,\text{R}} - \Lambda_t| \ge  \log
|\hat{Y}'_{m,\text{R}} - \Re\{ \Lambda_t \}|$,  the
application of Lemma~\ref{lemma:chi2} with $k=1$, and the independence between
$\hat{Y}'_{n,\text{R}}$ and $(\hat{Y}^{m-1}, \hat{Y}'_{n,\text{I}})$. 
Thus, taking expectation over $(\hat{Y}^{m-1}, \hat{Y}'_{n,\text{I}})$, we have
\begin{align}
  {\mathbb{E}\left[ \log |\det(\rvMat{J})| \right]} 
  &\ge \mathbb{E} \bigl[ \log |\det(\rvMat{N}_{\text{I}})| \bigr] + \constH.
\end{align}
It remains to show that $\mathbb{E} \bigl[ \log
|\det(\rvMat{N}_{\text{I}})| \bigr] > -\infty$. Let us recall that
$\rvMat{N}_{\text{I}} \defeq \Im\{ \diag\{\pmb{b}\}
\pmb{B}^*  \diag\{(\hat{Y}^{m-1})^*\} \} =
\pmb{T}_{\text{I}} \diag\{ \hat{Y}_{\text{R}}^{m-1}\} -  \pmb{T}_{\text{R}}  \diag\{ 
\hat{Y}_{\text{I}}^{m-1} \}$ where $\pmb{T}_{\text{R}}$ and $ \pmb{T}_{\text{I}}$ are the real and imaginary
parts of $\diag\{\pmb{b}\} \pmb{B}^*$, respectively, and thus depends only on $\pmb{H}$. Then 
\begin{align}
  \MoveEqLeft[0]{ \mathbb{E} \bigl[ \log
  |\det(\rvMat{N}_{\text{I}})| \bigr] } - \log |\det(\pmb{T}_{\text{I}})| \nonumber \\
  &= \mathbb{E} \left[ \log \Bigl|\det\bigl( 
  \diag\{  \hat{Y}_{\text{R}}^{m-1}\} -  \pmb{T}_{\text{I}}^{-1} \pmb{T}_{\text{R}} \diag\{
  \hat{Y}_{\text{I}}^{m-1} \} \bigr) \Bigr| \right]. \label{eq:tmp9000}
\end{align}%
We shall show that \eqref{eq:tmp9000} is lower-bounded as follows.
Note that $\hat{Y}_{m-1, \text{R}}$
can be written as $\hat{Y}_{m-1, \text{R}} = \hat{Y}'_{m-1, \text{R}} +
L_{m-1}(\hat{Y}_{\text{R}}^{m-2},\hat{Y}_\text{I}^{m-1})$ with $\hat{Y}'_{m-1, \text{R}}$ centered normal with non-zero bounded variance
and being independent of
$(\hat{Y}_{\text{R}}^{m-2},\hat{Y}_\text{I}^{m-1})$, and $L_{m-1}$ some linear operator that depends
only on $\pmb{H}$. Let
$\Pm_{m-1} \defeq \pmb{T}_{\text{I}}^{-1}
\pmb{T}_{\text{R}}$. It is easy to verify
that $\det \bigl( \diag\{ \hat{Y}_{\text{R}}^{m-1}\} - \pmb{P}_{m-1} \diag\{ \hat{Y}_{\text{I}}^{m-1} \}
\bigr)$ can be written as $ \det\bigl(\diag\{ \hat{Y}_{\text{R}}^{m-2}\} - \pmb{P}_{m-2}
\diag\{ \hat{Y}_{\text{I}}^{m-2} \}\bigr)
\hat{Y}'_{m-1, \text{R}} +
L_{m-1}'(\hat{Y}_{\text{R}}^{m-2},\hat{Y}_\text{I}^{m-1})$
where $\Pm_{m-2}$ is the upper-left $(m-2)\times
(m-2)$ sub-matrix of $\Pm_{m-1}$ and
$L_{m-1}'(\hat{Y}_{\text{R}}^{m-2},\hat{Y}_\text{I}^{m-1})$
is some value that is independent of $\hat{Y}'_{m-1,
\text{R}}$. Thus, we can again apply
Lemma~\ref{lemma:chi2} and obtain recursively
\begin{align}
  \MoveEqLeft[0]{\mathbb{E} \left[ \log \Bigl| \det\bigl( \diag\{ \hat{Y}_{\text{R}}^{m-1}\} - \pmb{P}_{m-1} \diag\{
  \hat{Y}_{\text{I}}^{m-1} \} \bigr) \Bigr| \right] } \nonumber \\
   &\ge \mathbb{E} \left[ \log \left| \det\Bigl(\diag\{ \hat{Y}_{\text{R}}^{m-2}\} - \pmb{P}_{m-2}
   \diag\{ \hat{Y}_{\text{I}}^{m-2} \}\Bigr) \hat{Y}'_{m-1, \text{R}}\right| \right] \\
   &\ \, \vdots \nonumber \\
   &\ge \mathbb{E} \left[ \log \left|  \hat{Y}_{1,\text{R}} - {P}_{1,1} 
   \hat{Y}_{1,\text{I}}  \right| \right] + \sum_{k=1}^{m-1} \E\bigl[ \log | \hat{Y}'_{k, \text{R}} | \bigr]\\
   &\ge \constH, 
\end{align}%
where the last inequality is from $\mathbb{E} \bigl[ \log |  \hat{Y}_{1,\text{R}} - {P}_{1,1}
\hat{Y}_{1,\text{I}}  | \bigr]  \ge \mathbb{E} \bigl[ \log |  \hat{Y}_{1,\text{R}} | \bigr] \ge
\constH$ due to the independence between $\hat{Y}_{1,\text{R}}$ and $\hat{Y}_{1,\text{I}}$ and the
application of Lemma~\ref{lemma:chi2}. 

Finally, recalling that $\pmb{T}_{\text{I}} \defeq \Im\{\diag\{\pmb{b}\} \pmb{B}^*\}$, we have $\log |\det(\pmb{T}_{\text{I}})|>-\infty$ for
any generic $\Hm$, it follows from \eqref{eq:tmp9000}
that $\E \bigl[ \log |\det(\rvMat{N}_{\text{I}})| \bigr]$ is lower-bounded. By now, we have shown
that $h(\rvVec{Y}) \ge n \log P + \constH$. Since $h(\rvVec{Y}) \ge h(\rvVec{Y}\cond
\tilde{\rvVec{Y}}) = h(\rvVec{Z}) \ge n\log(\pi e) \ge  0$, we have $h(\rvVec{Y}) \ge (n \log P +
\constH)^+ \ge n \log^+\! P - |\constH|$ from \eqref{eq:logAX} in Lemma~\ref{lemma:log+} with
$p=\frac{\constH}{n}$. This completes the proof for the case $n=2m-1$.

\subsubsection{Case $n\ne2m-1$}

Note that if \eqref{eq:hY} holds for $n=2m-1$, then it also holds for $n<2m-1$ and $n>2m-1$. To see
this, in the case with $n<2m-1$, we can add $2m-1-n$ receive antennas to have $(\rvVec{Y},
\rvVec{Y}')$ with $\rvVec{Y}'$ being the extra outputs. Since \eqref{eq:hY} holds for
$h(\rvVec{Y}, \rvVec{Y}')$ by assumption, then we have
\begin{align}
  h(\rvVec{Y}) &\ge  h(\rvVec{Y},\rvVec{Y}') - h(\rvVec{Y}') \\
  &\ge (2m-1) \log^+\! P - (2m-1-n) \log^+\! P + \constH \\
  &= n \log^+\! P + \constH,
\end{align}%
where the second inequality is from \eqref{eq:hY} and the fact that $h(\rvVec{Y}') \le (2m-1-n)
\log^+\!
P + \constH$. When $n>2m-1$, we partition $\rvVec{Y} = (\rvVec{Y}', \rvVec{Y}'')$ with
$\rvVec{Y}'\in\mathbb{C}^{(2m-1) \times 1}$ and obtain
\begin{align}
  h(\rvVec{Y}) &= h(\rvVec{Y}') + \sum_{k=2m}^n h(Y_k \cond Y^{k-1}) \\
  &\ge (2m-1) \log^+\! P \nonumber \\
  &\quad + \sum_{k=2m}^n h(Y_k \cond Y^{k-1}, \tilde{\Theta}_{\text{R}}^{k-1},
  \tilde{\rvVec{\Theta}}_{\text{T}}, \rvVec{X}) + \constH\\ 
   &\ge (2m-1) \log^+\!  P \nonumber \\
   &\quad + \sum_{k=2m}^n \E \left[ \log^+\! \biggl|
   \sqrt{\frac{P}{m}}\pmb{h}_{k}^\T
   (e^{j\tilde{\rvVec{\Theta}}_{\text{T}}} \circ \rvVec{X}_0 ) \biggr| \right] + \constH' \label{eq:tmp9011}\\ 
  &= \Bigl( 2m-1 + \frac{n-2m+1}{2} \Bigr) \log^+\! P + \constH'' \\
  &= \Bigl( m + \frac{n-1}{2} \Bigr) \log^+\! P + \constH'', 
\end{align}%
where the second inequality is from Lemma~\ref{lemma:ejt}; the equality \eqref{eq:tmp9011} is from
the fact that $\pmb{h}_{k}^\T (e^{j\tilde{\rvVec{\Theta}}_{\text{T}}} \circ \rvVec{X}_0)\sim
\mathcal{CN}(0,\|\pmb{h}_{k}\|^2)$.

\section*{Acknowledgement}
S. Yang would like to thank G. Durisi for helpful discussions and comments during the early
stage of this work.


\begin{IEEEbiographynophoto}{Sheng Yang} (M'07) 
 received the B.E.~degree in electrical engineering from Jiaotong University, Shanghai, China, in
 2001, and both the engineer degree and the M.Sc.~degree in electrical engineering from Telecom
 ParisTech, Paris, France, in 2004, respectively. In 2007, he obtained his Ph.D.~from Universit\'e de
 Pierre et Marie Curie (Paris VI). From October 2007 to November 2008, he was with Motorola Research
 Center in Gif-sur-Yvette, France, as a senior staff research engineer. Since December 2008, he has
 joined CentraleSup\'elec where he is currently an associate professor. From April 2015, he also holds an honorary faculty position in the department of electrical and electronic engineering of the University of Hong Kong~(HKU). He received the 2015 IEEE ComSoc Young Researcher Award for the Europe, Middle East, and Africa Region~(EMEA). He is an editor of the IEEE TRANSACTIONS ON WIRELESS COMMUNICATIONS. 
 \end{IEEEbiographynophoto}

\begin{IEEEbiographynophoto}{Shlomo Shamai~(Shitz)}~(F'94) received the B.Sc., M.Sc., and Ph.D.~degrees in electrical engineering from the Technion-Israel Institute of Technology, in 1975, 1981 and 1986 respectively.

During 1975-1985 he was with the Communications Research Labs, in the capacity of a Senior Research
Engineer. Since 1986 he is with the Department of Electrical Engineering, Technion-Israel Institute of Technology, where he is now a Technion Distinguished Professor, and holds the William Fondiller Chair of Telecommunications. His research interests encompasses a wide spectrum of topics in information theory and statistical communications.

Dr. Shamai (Shitz) is an IEEE Fellow, a member of the Israeli Academy of Sciences and Humanities and a foreign member of the US National Academy of Engineering. He is the recipient of the 2011 Claude E.~Shannon Award and the 2014 Rothschild Prize in Mathematics/Computer Sciences and Engineering.

He has been awarded the 1999 van der Pol Gold Medal of the Union Radio Scientifique Internationale
(URSI), and is a co-recipient of the 2000 IEEE Donald G.~Fink Prize Paper Award, the 2003, and the
2004 joint IT/COM societies paper award, the 2007 IEEE Information Theory Society Paper Award, the
2009 and 2015 European Commission FP7, Network of Excellence in Wireless COMmunications~(NEWCOM++,
NEWCOM\#) Best Paper Awards, the 2010 Thomson Reuters Award for International Excellence in Scientific Research, the 2014 EURASIP Best Paper Award~(for the EURASIP Journal on Wireless Communications and Networking), and the 2015 IEEE Communications Society Best Tutorial Paper Award. He is also the recipient of 1985 Alon Grant for distinguished young scientists and the 2000 Technion Henry Taub Prize for Excellence in Research. He has served as Associate Editor for the Shannon Theory of the IEEE TRANSACTIONS ON INFORMATION THEORY, and has also served twice on the Board of Governors of the Information Theory Society. He has served on the Executive Editorial Board of the IEEE TRANSACTIONS ON INFORMATION THEORY.
\end{IEEEbiographynophoto}

\end{document}